\documentclass[a4paper,11pt,reqno]{article}
\usepackage[title]{appendix}%
\usepackage[utf8]{inputenc}
\usepackage{amsmath}
\usepackage{amssymb}
\usepackage{mathrsfs}
\usepackage{mathtools}
\usepackage{enumerate}
\usepackage{enumitem}
\usepackage{amsthm}
\usepackage{booktabs}
\usepackage{xcolor}
\usepackage{url,hyperref}
\usepackage{multirow}
\allowdisplaybreaks
\usepackage{caption}
\usepackage{blindtext}
\usepackage{algpseudocode,algorithm}
\usepackage{bold-extra}

\usepackage{setspace}
\newcommand{\email}[1]{\href{mailto:#1}{\texttt{#1}}}

\DeclareMathSymbol{:}{\mathord}{operators}{"3A}

\usepackage{float}

\newcommand{\F}{\mathbb{F}}
\newcommand{\E}{\mathbb E}

\newcommand{\rk}{\mathrm{rk}}

\newcommand{\supp}{\mathrm{supp}}
\newcommand{\C}{\mathcal{C}}
\newcommand{\D}{\mathcal{D}}

\newcommand{\mcX}{\mathcal X}
\newcommand{\rowsp}{\mathrm{rowsp}}

\newcommand{\im}{\mathrm{im}}

\renewcommand{\epsilon}{\varepsilon}
\newcommand{\NN}{\mathbb N}

\newcommand{\Bil}{\mathrm{Bil}}
\newcommand{\mcB}{\mathcal B}
\newcommand{\qbinom}[2]{\genfrac[]{0pt}{0}{#1}{#2}}
\renewcommand{\Pr}{\mathbb{P}}

\usepackage{rotating}

\newcommand\restr[2]{{
  \left.\kern-\nulldelimiterspace 
  #1 
  \vphantom{\big|} 
  \right|_{#2} 
  }}

\usepackage[margin=2.7cm]{geometry} 
\usepackage{tabularx}

\makeatletter
\renewcommand\@makefnmark{}
\makeatother

\theoremstyle{plain} \numberwithin{equation}{section}
\newtheorem{theorem}{Theorem}[section]
\newtheorem{corollary}[theorem]{Corollary}

\newtheorem{Lemma}[theorem]{Lemma}
\newtheorem{proposition}[theorem]{Proposition}

\theoremstyle{definition}
\newtheorem{definition}[theorem]{Definition}

\newtheorem{remark}[theorem]{Remark}
\newtheorem{example}[theorem]{Example}

\newtheorem{notation}[theorem]{Notation}
\newtheorem{proposition/definition}[theorem]{Proposition/Definition}

\renewcommand{\phi}{\varphi}

\title{\textbf{The Star Product of Uniformly Random Codes}}

\author{Johan Vester Dinesen\thanks{Department of Mathematics and Systems Analysis, Aalto University, Espoo, Finland (\email{johan.v.dinesen@aalto.fi},\email{ragnar.freij@aalto.fi},\email{camilla.hollanti@aalto.fi})} \and Ragnar Freij-Hollanti$^*$ \and  Camilla Hollanti$^*$ \and Benjamin Jany\thanks{Department of Mathematics and Computer Science, Eindhoven University of Technology, Eindhoven, The Netherlands (\email{b.jany@tue.nl}, \email{a.ravagnani@tue.nl})} \and Alberto Ravagnani$^\dagger$}
\date{}

\begin{document}
\footnotetext{\textbf{Funding:}  J.V.~Dinesen was supported by the European Union MSCA Doctoral Networks (HORIZON-MSCA-2021-DN-01, Project 101072316). C. Hollanti was supported by the Business Finland Co-Innovation Consortium (Grant No. BFRK/473/31/2024, 5845483) and by the TUM Global Visiting Professor Program as part of the Excellence Strategy of the federal and state governments of Germany. B.~Jany was supported by the Casimir Institute at the Eindhoven University of Technology and by the Dutch Research Council via grant VI.Vidi.203.045. A.~Ravagnani was  supported by the Dutch Research Council via grant VI.Vidi.203.045.
}
\maketitle

\abstract{{
We consider the problem of determining the expected dimension of the star product of two uniformly random linear codes that are not necessarily of the same dimension. We use a correspondence between the star product and the evaluation of bilinear forms to provide an explicit lower bound on the expected star product dimension. We prove that the expected dimension asymptotically reaches its maximum possible value as the field size increases. Furthermore, we show that the same maximal dimension is achieved asymptotically as the code dimensions increase, subject to a condition bounding their relative growth rates. We also analyze the variance of the star product dimension, providing explicit asymptotic upper bounds. Finally, we discuss the implications of these results for private information retrieval, secure distributed matrix multiplication, quantum error correction, and cryptanalysis.}}

\section{Introduction}
The \textit{star product}, also called the \textit{Schur} or \textit{Hadamard product}, between subspaces $U,V\leq \F_q^n$ is the subspace spanned by the component-wise products of vectors from $U$ and $V$; see~\cite{PELLIKAAN1992369}. This operation is attracting increasing attention because of its connections with several areas of coding theory, information theory, and cryptography.

In the context of quantum error correction, star products have been used to construct stabilizer codes \cite{CSST,rengaswamy2020optimality,bodur2025schur} and to give an algebraic characterization of binary CSS-T codes \cite{camps2023algebraic}. In secure distributed matrix multiplication (SDMM) and private information retrieval (PIR), the algebraic structure of star products controls communication complexity and security guarantees \cite{PIR2017,SDMM2024}. The star product also appears naturally in secret sharing and secure multiparty computation \cite{Cramer1,MPC}, and  in linear exact repair schemes for distributed storage \cite{guruswamiwooters,matthewsrepair}. In code-based cryptography, star products appear in attacks through distinguishers \cite{10.1007/s10623-014-9967-z,Hormann2023Distinguishing,Mora2022dual}, mostly via squares of codes, i.e., a star product of a code with itself. These rely on certain structured codes like generalized Reed-Solomon (GRS) codes  typically having a much lower dimension than random codes when starred with themselves.  A well-known exception to this is the (binary) Goppa code that can withstand such distinguishing attacks with carefully chosen parameters \cite{Mora2022dual}.

The properties of star products have been studied in various algebraic settings \cite{randriambololona:hal-02287120}. For special classes of codes, particularly evaluation codes such as Reed-Solomon codes and algebraic geometry codes, star products are well understood and have proven useful in the construction of decoding algorithms \cite{PELLIKAAN1996229,7942048}. A notable contribution in this direction is the analysis of the square of a random linear code, $\mathcal{C}^2 = \mathcal{C} \star \mathcal{C}$ \cite{cascudo2015squares}, where asymptotic estimates of its dimension were derived. Alongside this, earlier work \cite{7282444} investigated the dimension of the star product of two random codes by analyzing the linear independence of rank 1 matrices, establishing that the expected dimension equals $\min\{n, k_1k_2\}$ with high probability for certain restricted parameter regimes.

{ 
In this paper, we study the expected dimension of the star product of two independent, uniformly random codes. More precisely, we investigate the typical behavior of $\dim (\mathcal{C}_1\star\mathcal{C}_2)$, where $\mathcal{C}_1,\mathcal{C}_2\leq \mathbb{F}_q^n$ are independent, uniformly random linear codes of prescribed dimensions. Extending the ideas of \cite{cascudo2015squares}, our approach establishes an explicit formula for the expected size of the kernel of a linear evaluation map whose image is exactly the star product of the two codes. While the approach in \cite{7282444} estimates the expected dimension by bounding the probability that random rank-1 matrices fail to achieve maximal span, our alternate perspective yields an explicit lower bound on $\mathbb{E}[\dim (\mathcal{C}_1\star\mathcal{C}_2)]$ for all parameter choices. We then determine the asymptotics of this bound both as the field size $q$ and the dimensions of $\mathcal{C}_1$ and $\mathcal{C}_2$ grow. Crucially, our approach yields asymptotic results in the field size that are independent of the relationship between the dimensions of the two codes. Furthermore, we leverage these bounds to upper bound the variance of the star product dimension, showing that it decays asymptotically to $0$ in certain parameter regimes. We provide a detailed technical comparison between our results and the results of \cite{7282444} in Remark \ref{remark:otherpaper}.
}

In the second part of the paper, we discuss the implications of these results for private information retrieval applications and secure distributed matrix multiplication, as well as the potential of the results obtained to be utilized in distinguisher-type attacks in code-based cryptography.

The remainder of this paper is outlined as follows. Section~\ref{sec:prelim} recalls the coding theory concepts needed throughout the paper. Section~\ref{sec:bil} develops the connection between star products and bilinear forms. In Section~\ref{sec:asympq} and Section~\ref{sec:asympindim}, we analyze the asymptotic behavior of the star product, first as the field size grows and then as the code dimensions increase under specific admissibility conditions. In both cases, we conclude that the expected dimension attains its maximum value and that the variance vanishes asymptotically. Finally, in Section~\ref{sec:app} we discuss the implications of these results for certain applications.

\section{Preliminaries on linear codes}
\label{sec:prelim}

Throughout this paper, $\F$ denotes a field, $q$ a prime power, and $\F_q$ the finite field of $q$ elements. We denote subsets by $\subseteq$ and subspaces by $\leq$. Vectors are row vectors unless otherwise specified. 
We let $A^{(i)}$ denote the $i$th column
of a matrix $A$. The natural numbers are  $\NN = \{0,1,2,\ldots\}$. For $n,k\in \NN$, the \emph{$q$-binomial coefficient} is
\[
\qbinom{n}{k}_q = \frac{(1-q^n)(1-q^{n-1})\ldots(1-q^{n-k+1})}{(1-q)(1-q^2)\ldots(1-q^k)},
\]
which evaluates to 0 when $k>n$. When both the numerator and denominator are empty products, we set the quantity to 1. We will also allow negative inputs for both $q$-binomials and binomials, but in that case they are defined to be 0.

We start by recalling some notions from coding theory. For  general reference, we refer the reader to~\cite{vanLint}. \\

\begin{definition} A \emph{linear code} is a non-zero $\F_q$-linear subspace $\C\leq \F_q^n$. The \emph{minimum Hamming distance} of $\C$ is 
$d(\C) = \min\{ w(c) \mid c\in \C, \; c\neq 0\},$
where $w$ denotes the Hamming weight on $\F_q^n$.
\end{definition}

For the remainder of this paper, all codes are assumed to be linear. We recall some other coding theory concepts in passing.

The \emph{support} of a code $\C$ is $\supp(\C) = \{ i\in [n] \mid \text{ there exists } c\in \C \text{ with } c_i \neq 0\}.$ $\C$ is called \emph{degenerate} if $\supp (\C) \neq [n]$. For $I\subseteq [n]$, $\pi_I(\C)$ is the code obtained by projecting $\C$ onto the coordinates indexed by $I$, that is, $\pi_I(\C) = \{(c_i)_{i\in I}\mid c\in \C\}$.

A \emph{generator matrix} of a code $\C\leq \F_q^n$ with dimension $\dim \C = k$ is a $k\times n$ matrix $G$ over~$\F_q$ such that $\rowsp \,G = \C$. A generator matrix is in \emph{systematic form} if $G=[I_k \mid A]$ for some matrix $A$. A \emph{parity-check matrix} of $\C$ is a matrix $H\in \F_q^{(n-k)\times n}$ such that $\C = \ker H^\top$. We say that codes $\C_1,\C_2\leq \F_q^n$ are \emph{monomially equivalent} if there exists a monomial $n\times n$ matrix~$M$ such that if $G_1$ is a generator matrix of $\C_1$, then $G_1M$ is a generator matrix of $\C_2$.

All codes $\C\leq \F_q^n$ satisfy the \emph{Singleton bound}: $\dim\C \leq n - d(\C) +1$. Codes that meet the Singleton bound with equality are called \textit{Maximum Distance Separable} (\textit{MDS}) codes. Note that MDS codes are always non-degenerate.

The key player of this paper is the star product.

\begin{definition}
   The \emph{star product} of
vectors $x,y\in \F_q^n$ is their component-wise multiplication, denoted by $x\star y = (x_1,\ldots,x_n) \star (y_1,\ldots,y_n) = (x_1y_1,\ldots,x_ny_n)$. 
The  \emph{star product}
of linear codes $\C_1,\C_2\leq \F_q^n$
is the $\F_q$-linear span of the pairwise star products of two codewords, that is, $$\C_1 \star \C_2 =\langle x\star y\mid x\in \C_1,y\in \C_2\rangle_{\F_q}.$$
\end{definition}

The star product is also called
\emph{Schur product} or \emph{Hadamard product}. 
In general, it is difficult to determine the dimension of the star product of two codes. A trivial upper bound is $\dim (\C_1\star\C_2)\leq \min\{\dim \C_1 \dim \C_2,n\}$ since fixing bases of $\C_1$ and $\C_2$, and then taking all pairwise star products between these basis vectors yields a spanning set for $\C_1 \star \C_2$ (and in general not a basis). 
Other bounds were derived with additional assumptions on either of the two codes.

\begin{Lemma}[\cite{Singletonupper}, Lemma 5] \label{lemma:singletonupper}
Let $\C_1,\C_2 \leq \F_q^n$ be non-degenerate codes. Then,
\[
    \dim (\C_1 \star \C_2) \geq \min\{n,\dim \C_1+ d(\C_2^\perp)-2, \dim \C_2 + d(\C_1^{\perp}) -2\}.
\]
\end{Lemma}

This lower bound highlights the role of the dual distance of one code in affecting the dimension of the star product. In the special case where one of the two codes is MDS, the dual distance attains its maximum possible value, allowing for a stronger bound. Since the dual of an MDS code is also MDS, the following result holds.

\begin{proposition}[\cite{randriambololona:hal-02287120}]\label{prop:MDSnondeg}
    Let $\C_1,\C_2 \leq \F_q^n$ be non-degenerate codes. If at least one of the codes is MDS, then
    \[\dim (\C_1\star\C_2) \geq \min \{ n, \dim \C_1 + \dim \C_2 -1 \}.\]
\end{proposition}

For sufficiently large field sizes, we expect the conditions of Proposition~\ref{prop:MDSnondeg} to hold with high probability, since random codes are then likely to be MDS. This result will be necessary to prove Proposition~\ref{prop:expdim-mds-rnd}.

\section{Star product and bilinear forms}%
\label{sec:bil}

Star products and bilinear forms are closely related to one another, as shown in \cite{randriambololona:hal-02287120}. In this section, we show that the star product of two codes, denoted by $\C_1 \star \C_2$, is the evaluation of all bilinear forms at points corresponding to the columns of the generator matrices of~$\C_1$ and~$\C_2$. To make this relation precise, we first recall the definition of a bilinear form.

\begin{definition}
    Let $V$ and $W$ be $\F$-vector spaces. A map $\varphi \, :\,  V \times W \rightarrow \F$ is a bilinear form if for all $v_1, v_2 \in V$,  $w_1, w_2 \in W$, and $\lambda \in \F$, the following  hold:
    \begin{itemize}
        \item[1)] $\varphi(v_1 + v_2, w_1) = \varphi(v_1, w_1) + \varphi(v_2, w_1)$,
        \item[2)] $\varphi(v_1 , w_1 + w_2) = \varphi(v_1, w_1) + \varphi(v_1, w_2)$,
        \item[3)] $\varphi(\lambda v_1, w_1) = \varphi(v_1, \lambda w_1) = \lambda \varphi(v_1, w_1)$.
    \end{itemize}
We let $\Bil(V,W)$ denote the set of  bilinear forms on $V \times W$. It is well known that $\Bil(V, W)$ is an $\F$-vector space of dimension $\dim(V)\dim(W)$.
\end{definition}

In the sequel we work with positive integers $k_1$ and $k_2$ and let $\mathcal{B} \coloneqq \Bil(\F_q^{k_1},\F_q^{k_2})$. Recall that there is a one-to-one correspondence between bilinear forms $\varphi \in \mathcal{B}$ and matrices $B \in \F_q^{k_1 \times k_2}$. We now state the relation between star products and bilinear forms, which was established in~\cite{randriambololona:hal-02287120}. We provide a proof for completeness.

\begin{theorem}[\cite{randriambololona:hal-02287120}]\label{thm:starbil}
    Let $\C_1, \C_2 \leq \F_q^n$ with generator matrices $G_1 \in \F_q^{k_1 \times n}$ and $G_2 \in \F_q^{k_2 \times n}$. Define the linear map
    $$\psi_{G_1, G_2} \, : \, \mathcal{B} \rightarrow \F_q^n, \quad \psi_{G_1, G_2}(\varphi) = \left(
        \varphi\left(G_1^{(1)}, G_2^{(1)}\right),\ldots, \varphi\left(G_1^{(n)}, G_2^{(n)}\right)
    \right).$$
Then, $\C_1 \star \C_2 = \im(\psi_{G_1,G_2})$ and  $$\dim (\C_1\star\C_2) = k_1 k_2 - \log_q (|\ker \psi_{G_1,G_2}|).$$
\end{theorem}

\begin{proof}
It is immediate from the definition that $\psi_{G_1,G_2}$ is an $\F_q$-linear map. Furthermore, $\mathcal{B}$ is generated by the bilinear forms $\Omega=\{\varphi_{ij} \, : \,  1 \leq i \leq  k_1, \, 1 \leq j \leq k_2\}$, where $\varphi_{ij}(v,w) = v_iw_j$. Let
\[
    G_1 = \begin{bmatrix} \alpha_1 \\ \alpha_2 \\ \vdots \\ \alpha_{k_1} \end{bmatrix},\quad G_2 = \begin{bmatrix} \beta_1 \\ \beta_2 \\\vdots \\ \beta_{k_2} \end{bmatrix},
\]
where $\alpha_i$ and $\beta_j$ denote the $i$th row of $G_1$ and $j$th row of $G_2$, respectively. Then for any $\varphi_{i,j}\in \Omega$, we have
\begin{align}
\psi_{G_1, G_2}(\varphi_{ij}) = \begin{bmatrix}
       \alpha_{i,1} \beta_{j,1} &  \alpha_{i,2}\beta_{j,2} & \ldots &  \alpha_{i,n}\beta_{i,n}
   \end{bmatrix} = \alpha_i \star \beta_j \in \C_1 \star \C_2, \label{eq:correspondence1}
\end{align}
so that $\im(\psi_{G_1,G_2}) \subseteq \C_1 \star \C_2$.

Now consider $c\in \C_1 \star \C_2$. Let $\alpha\in \C_1$ and $\beta\in \C_2$ be such that $c=\alpha \star \beta$. Then, for some scalars $\lambda_{i,j}\in \F_q$, we can write
\begin{align}
    c &= \sum_{1 \leq i \leq k_1} \sum_{1 \leq j \leq k_2} \lambda_{ij} (\alpha_i \star \beta_j ) \notag \\
                &= \sum_{1 \leq i \leq k_1} \sum_{1 \leq j \leq k_2} \lambda_{ij} \psi_{G_1,G_2}(\varphi_{ij}) \label{eq:correspondence2}\\
                &= \psi_{G_1,G_2}\left(\sum_{1 \leq i \leq k_1} \sum_{1 \leq j \leq k_2} \lambda_{ij} \varphi_{ij}\right),\label{eq:correspondence3}
\end{align}
where \eqref{eq:correspondence2} follows by \eqref{eq:correspondence1}, and \eqref{eq:correspondence3} follows by $\F_q$-linearity of $\psi_{G_1, G_2}$.
Hence $c \in \im(\psi_{G_1, G_2})$, which shows that $\C_1 \star \C_2 \subseteq \im(\psi_{G_1, G_2})$ and thus the two spaces are equal.
\end{proof}

\begin{remark}
   Note that  different pairs of generator matrices $(G_1, G_2)$ and $(G'_1, G'_2)$ for the codes $(\C_1, \C_2)$ induce in general different maps $\psi_{G_1, G_2}$ and $\psi_{G'_1,G'_2}$, as defined in Theorem \ref{thm:starbil}. However, the image of those two maps remains the same since they are both equal to $\C_1 \star\C_2$. Furthermore, if $\C_1$ is monomially equivalent to $\C_1'$, then $\C_1 \star \C_2$ is also monomially equivalent to $\C_1'\star \C_2$. Therefore, in the rest of the article we always choose a generator matrix in systematic form and denote the corresponding map introduced in Theorem \ref{thm:starbil} by $\psi_{\C_1, \C_2}$.
\end{remark}

By Theorem \ref{thm:starbil}, one can determine the dimension of the star product between two codes by counting the number of elements in the kernel of the map $\psi_{G_1, G_2}$, for some given generator matrices $G_1 \in \F_q^{k_1 \times n}$ and $ G_2 \in \F_q^{k_2 \times n}$ both in systematic form. To estimate this kernel, we can count the number of bilinear forms $\varphi\in \mathcal{B}$ of rank $r$ such that $\varphi(e_i^1,e_i^2)=0$ for $i\in [k_1]$, or equivalently the number of $k_1\times k_2$ matrices over $\F_q$ of rank $r$ with zero diagonal. This is a known problem in enumerative combinatorics~\cite{lewis2011matrices,ravagnani2018duality}.
We introduce notation that
will be used in the remainder of the article. 

\begin{definition}\label{not2}
Let $r,k_1,k_2 \in \mathbb{N}$ such that $r\leq k_1\leq k_2$, $I\subseteq K \coloneqq  [k_2]\setminus [k_1]$ and $\alpha\in \mathbb{R}$. We then define the following sets:
\begin{align*}
 S_r^{k_1,k_2} &\coloneqq \left\{A\in \F_q^{k_1\times k_2} \,\middle|\, \rk(A) = r, \; A_{i,i}=0 \text{ for all } i\in [k_1] \right\}, \\
    S_{r,I}^{k_1,k_2} &\coloneqq \left\{A \in S_{r}^{k_1,k_2}\, \middle|\, A^{(i)}=0\text{ for all } i\in I, A^{(\ell)}\neq 0\text{ for all } \ell\in K\setminus I \right\},\\
    S^{k_1,k_2} &\coloneqq \bigcup_{r=0}^{k_1}S_r^{k_1,k_2}.
\end{align*}
Since $|S_{r,I}^{k_1,k_2}| = |S_{r,J}^{k_1,k_2}|$ for $I,J\subseteq K$ when $|I|=|J|$ we let $S_{r,\ell}^{k_1,k_2} \coloneqq |S_{r,I}^{k_1,k_2}|$, where $I\subseteq K$ with $|I|=\ell$. The subscripts and superscripts may be omitted if they are clear from context.
\end{definition}

\begin{Lemma} \label{Lemma:vanishmaindiagonal}
Let $r,k_1,k_2\in \mathbb{N}$ such that $r\leq k_1\leq k_2$ and $I\subseteq K$. Let
\[
    \theta(i,j) = \binom{k_1}{i}(q-1)^i(-1)^{r-j}q^{jk_2 + \binom{r-j}{2}}.
\]
Then 

\[
|S^{k_1,k_2}_r| = q^{-k_1}\sum^{k_1}_{i=0} \left(\sum^{k_1}_{j=0}\theta(i,j)\qbinom{k_1-i}{j}_q\qbinom{k_1-j}{k_1-r}_q \right),
\]
and
\[
|S^{k_1,k_2}_{r,I}| =\sum_{I \subseteq J \subseteq K} (-1)^{|J| - |I|} q^{-k_1}\Bigg(\sum^{k_1}_{i=0} \Bigg(\sum^{k_1}_{j=0} \theta(i,j)q^{-j|J|} \qbinom{k_1-i}{j}_q\qbinom{k_1-j}{k_1-r}_q \Bigg)\Bigg).
\]

\end{Lemma}
\begin{proof}
The first part of the statement is~\cite[Corollary~60]{ravagnani2018duality}.
For the second part, define
\begin{align*}
    f(I) \coloneqq |S^{k_1,k_2}_{r,I}|, \qquad
    g(I) \coloneqq \left|\left\{ A\in S_r^{k_1,k_2} \,\middle|\,  A^{(i)} = 0 \text{ for all } i\in I\right\}\right|.
\end{align*}
Clearly $g(I) = \sum_{I\subseteq J \subseteq K} f(J)$ for any $I\subseteq K$. Hence, by Möbius inversion we obtain
\[
f(I) = \sum_{I\subseteq J \subseteq K} (-1)^{|J|-|I|}g(J).
\]
Using the first part of the statement we determine $g(J)$ for any $J\subseteq K$ as \\
\begin{align*}
    g(J) &= |S_r^{k_1,k_2-|J|}|\\  
    &= q^{-k_1} \sum^{k_1}_{j=0}\binom{k_1}{j}(q-1)^j\left( \sum^{k_1}_{l=0}(-1)^{r-l}q^{l(k_2-|J|) + \binom{r-l}{2}} \qbinom{k_1-j}{l}_q \qbinom{k_1-l}{k_1-r}_q \right),
\end{align*}
and the result follows.
\end{proof}

The next result counts the number of zeros of a bilinear form, defined as follows. 

\begin{definition}
    Let $\varphi \in \mathcal{B}$. 
    An element $(v_1, v_2) \in \F_q^{k_1} \times \F_q^{k_2}$ is \emph{a zero} of $\varphi$ if $\varphi(v_1,v_2) = 0$. 
      Denote by $Z(\varphi)$ the number of zeros of $\varphi$. 
\end{definition}

\begin{Lemma} \label{Lemma:zerosbyrank}
    Let $\varphi \in \mathcal{B}$ be of rank $r$. Then $Z(\varphi) = (q^r+q-1)q^{k_1+k_2-r-1}$ and $Z$ is a monotonically decreasing function in $r$.
\end{Lemma}

\begin{proof}
Let $A$ be the matrix representation of $\varphi$. Then $\rk(A) = r$. 
    \begin{align*}
        Z(\varphi) &= |\{(v_1, v_2)\in \F_q^{k_1}\times \F_q^{k_2} \, \colon \, v_1Av_2^\top = 0\}|\\
                    &= \sum_{v_1 \in \F_q^{k_1}} |\{v_2 \in \F_q^{k_2} \, : \, (v_1A)v_2^\top = 0\}| \\
                    &= \sum_{\substack{v_1 \in \F_q^{k_1}\\ v_1A = 0}}|\{v_2 \in \F_q^{k_2} \, : \, (v_1A)v_2^\top  = 0\}| + \sum_{\substack{v_1 \in \F_q^{k_1}\\ v_1A \neq 0}}|\{v_2 \in \F_q^{k_2} \, : \, (v_1A)v_2^\top  = 0\}|\\
                    &= \sum_{\substack{v_1 \in \F_q^{k_1}\\ v_1A = 0}} q^{k_2} + \sum_{\substack{v_1 \in \F_q^{k_1}\\ v_1A \neq 0}} q^{k_2-1}\\
                    &= q^{k_1 -r}q^{k_2} + (q^{k_1} - q^{k_1 -r})q^{k_2-1} = (q^r+q-1)q^{k_1+k_2-r-1}.
    \end{align*}
In the third equality $|\{v_2 \in \F_q^{k_2} \, : \, (v_1A)v_2^\top = 0\}| = q^{k_2 -1}$ for $v_1A \neq 0$, since the linear map $v_1A \, : \, \F_q^{k_2} \rightarrow \F_q$ with $v_2 \mapsto v_1Av_2^\top$ is surjective of rank $1$, and thus has a kernel of dimension $k_2 -1$. 
In the last equality, we have that $\{v_1 \in \F_q^{k_1} \, : \, v_1A = 0\}$ is the left kernel of $A$, which has dimension $k_1 -r$ since $A$ has rank $r$. Lastly, $Z$ is monotonically decreasing as $\frac{\mathrm{d}}{\mathrm{d}r} Z = -\ln(q) (q-1)q^{-r+k_2+k_1-1}$ is negative.
\end{proof}

We now proceed to compute the expected size of the kernel of the map $\psi_{G_1,G_2}$, where the generator matrices $G_1$ and $G_2$ are chosen at random according to the model defined below.

\begin{definition}\label{def:setuprand}
In this paper, a uniformly random code $\C \leq \F_q^n$ of dimension $k$ is constructed by choosing uniformly at random and independently the last $n-k$ columns of a systematic generator matrix.
\end{definition}

The authors of~\cite{cascudo2015squares} explicitly discuss alternative probabilistic models for choosing a random code and argue that their results apply more broadly than just the systematic generator matrix model used in Definition~\ref{def:setuprand}. Specifically, as noted in their Remark~2.1 and in the discussion leading to Section~6, they show that the main conclusions regarding the square of a random code remain essentially unchanged under other standard distributions, such as choosing codes or generator matrices uniformly at random. Since our analysis for pairs of random codes (i.e., counting bilinear forms) follows the same combinatorial and algebraic principles, their arguments translate directly to our setting. Consequently, our results exhibit the same robustness with respect to the choice of random model, and we therefore confine our proofs to the model in Definition~\ref{def:setuprand} for technical convenience. 

\begin{theorem} \label{theorem:expectedkernel}    Let $\C_1, \C_2 \leq \F_q^n$ be uniformly random codes of dimensions $k_1$ and $k_2$, respectively, as in Definition~\ref{def:setuprand}. Let $\gamma(j) = (q^j+q-1)q^{-j}$. Then the expected value of $|\ker \psi_{\C_1,\C_2}|$ is
\begin{align*} 
&\mathcolor{white}{=}\mathbb{E}\big[|\ker \psi_{\C_1,\C_2}|\big]\\ 
&= 
\sum_{r=0}^{k_1}\sum_{i=0}^{k_1}\sum_{j=0}^{r}(-1)^{r-j}\gamma(r)^{n-k_2}\gamma(j)^{k_2-k_1} (q-1)^i q^{jk_2-n+ \binom{r-j}{2}}\binom{k_1}{i}\qbinom{k_1-i}{j}_q \qbinom{k_1-j}{k_1-r}_q.
\end{align*}
\end{theorem}
\begin{proof}
    Let $G_1 \in \F_q^{k_1 \times n}$ and $G_2 \in \F_q^{k_2 \times n}$ be systematic generator matrices of $\C_1$ and $\C_2$, respectively. Let $\mathcal{X}_j$ denote the set of all $k_j\times n$ generator matrices in systematic form for $j=1,2$. For $\varphi\in \mcB$ and $(\pi_1,\pi_2)\in \mcX_1\times \mcX_2$ we use a slight abuse of notation and write $\varphi(\pi_1,\pi_2) =0$ to mean \smash{$\varphi\left(\pi_1^{(i)},\pi_2^{(i)}\right) = 0$} for all $i\in [n]$. Then 
\begin{align}
    &\mathcolor{white}{=} \mathbb{E}\big[|\ker(\psi_{\C_1,\C_2})|\big] 
     = \sum^{q^{k_1+k_2}}_{i=0} i \cdot \mathbb{P}(|\ker( \psi_{\C_1,\C_2})| = i) \notag \\
     &=\frac{1}{q^{(n-k_1)k_1+(n-k_2)k_2}}\sum_{i=1}^{q^{k_1+k_2}}i \left|\big\{(\pi_1, \pi_2) \in  \mcX_1 \times \mcX_2 \, \colon \,  | \{\varphi \in \mcB\, \colon \, \varphi(\pi_1, \pi_2) = 0\}|=i \big\} \right| \notag \\
     &= \frac{1}{q^{(n-k_1)k_1+(n-k_2)k_2}} \sum_{(\pi_1,\pi_2) \in \mcX_1 \times \mcX_2} | \{\varphi \in \mcB \, \colon \, \varphi(\pi_1, \pi_2) = 0\}| \notag \\
     &=  \frac{1}{q^{(n-k_1)k_1+(n-k_2)k_2}}\sum_{\varphi \in \mcB} |\{ (\pi_1, \pi_2) \in \mcX_1 \times \mcX_2 \, \colon \, \varphi(\pi_1, \pi_2) = 0\}| \notag \\
    &
=\frac{1}{q^{(n-k_1)k_1+(n-k_2)k_2}}
\sum_{r=0}^{k_1}\sum_{I \subseteq K}\sum_{A \in S_{r,I}}
\Bigg(
\left((q^r + q - 1)q^{k_1+k_2-r-1}\right)^{n-k_2} \notag \\
&\mathcolor{white}{=} \hspace{5.7cm}
q^{k_1|I|}(q^{k_1-1})^{k_2-k_1-|I|}
\Bigg)
\label{eq:proof1} \\
         &=  \frac{1}{q^{(n-k_1)k_1+(n-k_2)k_2}}\sum_{r=0}^{k_1}\sum_{I \subseteq K}\sum_{A \in S_{r,I}} \left((q^r +q-1)q^{k_1+k_2-r-1}\right)^{n-k_2} q^{(k_2-k_1)(k_1-1) +|I|}\notag \\
      &= \frac{q^{(k_1+k_2-1)(n-k_2)+(k_2-k_1)(k_1-1)}}{q^{(n-k_1)k_1+(n-k_2)k_2}} \sum_{r=0}^{k_1}\sum_{I \subseteq K} \left((q^r +q-1)q^{-r}\right)^{n-k_2} q^{|I|}|S_{r,I}| \notag \\
      &= q^{k_1-n}\sum_{r=0}^{k_1}\sum_{I \subseteq K} \left((q^r +q-1)q^{-r}\right)^{n-k_2} q^{|I|}\sum_{I \subseteq J \subseteq K} (-1)^{|J| - |I|}q^{-k_1} \notag \\ &\mathcolor{white}{=}\sum^{k_1}_{i=0}\binom{k_1}{i}(q-1)^i \left(\sum^{k_1}_{j=0}(-1)^{r-j}q^{j(k_2 - |J|) + \binom{r-j}{2}} \qbinom{k_1-i}{j}_q \qbinom{k_1-j}{k_1-r}_q \right) \notag \\
     &=q^{-n}\sum_{r=0}^{k_1}\sum_{i=0}^{k_1}\sum_{j=0}^{k_1} \Bigg(\left((q^r+q-1)q^{-r}\right)^{n-k_2}\binom{k_1}{i}(q-1)^i(-1)^{r-j}q^{jk_2 + \binom{r-j}{2}}  \notag\\
&\mathcolor{white}{=} \hspace{2.7cm}\qbinom{k_1-i}{j}_q\qbinom{k_1-j}{k_1-r}_q\sum_{ I \subseteq K} \sum_{I \subseteq J \subseteq K} (-1)^{|J|-|I|}q^{|I|}q^{-j|J|} \Bigg) \label{eq:proof2} ,
\end{align}
where \eqref{eq:proof1} follows by Lemma \ref{Lemma:zerosbyrank} and \eqref{eq:proof2} follows by Lemma \ref{Lemma:vanishmaindiagonal}. The last two summands of \eqref{eq:proof2} can be further simplified using the Binomial Theorem:
\begin{align*}
 \sum_{I \subseteq K} \sum_{I \subseteq J \subseteq K} (-1)^{|J|-|I|}q^{|I|}q^{-j|J|}
   &=\sum_{J \subseteq K} (-1)^{|J|}q^{-j|J|} \sum_{I \subseteq J} (-1)^{|I|} q^{|I|}\\
   &= \sum_{J \subseteq K} (-1)^{|J|}q^{-j|J|} \sum_{t=0}^{|J|} \binom{|J|}{t}(-q)^{t}\\
   &= \sum_{J \subseteq K} (-1)^{|J|}q^{-j|J|} (1-q)^{|J|}\\
   &= \sum_{t=0}^{k_2-k_1}\binom{k_2-k_1}{t}(-q^{-j} + q^{-j+1})^{t}\\
   &= ((q^j+q-1)q^{-j})^{k_2-k_1}.
\end{align*}
Thus, we have that
\begin{multline*}
 \mathbb{E}\big[|\ker(\psi_{\C_1,\C_2})|\big]  = q^{-n} \sum_{r=0}^{k_1}\sum_{i=0}^{k_1}\sum_{j=0}^{k_1}(-1)^{r-j}\left((q^r +q-1)q^{-r}\right)^{n-k_2}((q^j+q-1)q^{-j})^{k_2-k_1}  \\ (q-1)^iq^{jk_2 + \binom{r-j}{2}}\binom{k_1}{i}\qbinom{k_1-i}{j}_q \qbinom{k_1-j}{k_1-r}_q,
\end{multline*}
and the result follows.
\end{proof}

In the case where $\C_1$ and $\C_2$ have the same dimension, we get the following value for $\E[|\ker \psi_{\C_1, \C_2}|]$.

\begin{corollary}
       Let $\C_1, \C_2 \leq \F_q^n$ be uniformly random codes of dimension $k=k_1=k_2$ as in Definition~\ref{def:setuprand}. Let $\gamma(j) = (q^j + q - 1)q^{-j}$. Then 
    \begin{align*}&\E[|\ker \psi_{\C_1, \C_2}|]= \\
    &\: \frac{1}{q^{n}}\sum_{r=0}^{k}\sum^{k}_{i=0} \sum^{\min\{r, k-i\}}_{j=0} (-1)^{r-j} \gamma(r)^{n-k}(q-1)^i q^{jk + \binom{r-j}{2}} \binom{k}{i} \qbinom{k-i}{j}_q \qbinom{k-j}{k-r}_q .
\end{align*}
\end{corollary}

To conclude the section, we use Theorem \ref{theorem:expectedkernel} to give a lower bound on the expected dimension of two uniformly random codes of dimensions $k_1$ and $k_2$.  We also experimentally compare this bound to the approximate expected value, sampled by Monte-Carlo, for some parameters in Table \ref{tab:lowerboundexperiments}. See Appendix \ref{app:table} for a comparison of the lower bound of Corollary \ref{cor:lowerboundstar} to the expected dimension for values $k_1=2,3$; $k_2=3,4$; $n=7,11,15$; and $q=2,3,5,7$. As we can see, even for low values our bound is appears to be good, and approaches the approximate expected value as any of the parameters increases.

\begin{corollary} \label{cor:lowerboundstar}
    Let $\C_1, \C_2 \leq \F_q^n$ be uniformly random codes of dimensions $k_1$ and $k_2$, respectively, as in Definition~\ref{def:setuprand}. Let $\gamma(j) = (q^j+q-1)q^{-j}$ for $j\in \mathbb{N}$. Then
\begin{multline}
    \mathbb{E}[\dim (\C_1\star\C_2)] \geq k_1 k_2 - \log_q \Bigg( \sum_{r=0}^{k_1}\sum_{i=0}^{k_1}\sum_{j=0}^{k_1}(-1)^{r-j}\gamma(r)^{n-k_2}\gamma(j)^{k_2-k_1} (q-1)^i \\
    q^{jk_2-n+ \binom{r-j}{2}}\binom{k_1}{i}\qbinom{k_1-i}{j}_q \qbinom{k_1-j}{k_1-r}_q \Bigg).
\end{multline}
\end{corollary}
\begin{proof}
By Theorem~\ref{thm:starbil} we have $$\dim (\C_1\star\C_2) = k_1k_2 - \dim( \ker \psi_{\C_1,\C_2}) = k_1 k_2 - \log_q \left(|\ker \psi_{\C_1,\C_2}|\right).$$ We then have by Jensen's Inequality that 
$$\mathbb{E}[\dim (\C_1\star\C_2)] = k_1k_2 - \mathbb{E}\big[\log_q \left( |\ker \psi_{\C_1,\C_2}|\right)\big] \geq k_1k_2 - \log_q\left(\mathbb{E}\big[ |\ker \psi_{\C_1,\C_2}|\big]\right),$$ and the result then follows by Theorem \ref{theorem:expectedkernel}.
\end{proof}

\section{Asymptotic behavior as the field size grows}
\label{sec:asympq}

In this section, we establish that the expected dimension of the star product is maximal with respect to the upper bound $\dim(\C_1 \star \C_2) \leq \min\{n,\dim \C_1 \dim \C_2\}$, when we let the field size $q$ tend to infinity. We begin by computing the asymptotic behavior of the expected size of the kernel of the map $\psi_{\C_1, \C_2}$ for two random codes $\C_1, \C_2$ each of fixed dimension. Recall that functions $f,g \colon \NN \rightarrow \mathbb{R}$ are \emph{asymptotically equivalent}, denoted by $f\sim_x g$, if $\lim_{x\rightarrow \infty} \frac{f(x)}{g(x)}=1$; see for instance~\cite{debruijn}. {  Similarly, we use the subscript notation $\mathcal{O}_x(\cdot)$ and $o_x(\cdot)$ to denote the standard Big-$\mathcal{O}$ and little-$o$ asymptotic bounds taken strictly with respect to the limit $x \to \infty$, treating all other parameters as fixed constants.}

\begin{theorem} \label{thm:qlimit} 
    Let $\C_1, \C_2 \leq \F_q^n$ be uniformly random codes of dimensions $k_1$ and $k_2$, respectively, as in Definition~\ref{def:setuprand}. Then
$$\lim_{q \rightarrow \infty} \E[|\ker \psi_{\C_1, \C_2}|] = \lim_{q\rightarrow \infty} (1 + q^{k_1k_2-n}) = \begin{cases}
    1 &\text{ if } k_1k_2 < n,\\
    2 &\text{ if } k_1k_2 = n,\\
     \infty &\text{ if } k_1k_2 > n.
\end{cases}$$

\end{theorem}

\begin{proof}
By Theorem \ref{theorem:expectedkernel} we know that
$$\E[ |\ker \psi_{\C_1, \C_2}|] = q^{-n} \sum_{r=0}^{k_1} \gamma(r)^{n-k_2}\sum_{i=0}^{k_1} \sum_{j=0}^{\min\{r, k_1-i\}} (-1)^{r-j}\binom{k_1}{i}f(i,j,r),$$
where $\gamma(r) =(q^r+q-1)q^{-r}$ and 

$$f(i,j,r) = \gamma(j)^{k_2-k_1}(q-1)^iq^{jk_2 + \binom{r-j}{2}}\qbinom{k_1-i}{j}_q\qbinom{k_1-j}{k_1-r}_q.$$

We treat the summand for $r=0$ and the others differently.  For $r=0$ only $j=0$ contributes to the sum, so
\begin{align*}
   q^{-n} \gamma(0)^{n-k_2}\sum_{i=0}^{k_1} f(i,0) =  q^{-n} q^{n-k_2} q^{k_2 - k_1}\sum_{i=0}^{k_1} \binom{k_1}{i}(q-1)^i = 1.
\end{align*}
For $r \ge 1$, note that the $q$-binomial satisfies $\qbinom{n}{k}_q\sim_q q^{k(n-k)}$. Using this, one can verify that
$$f(i,j,r) \sim_q \begin{cases}
    q^{k_2 - k_1 + i - \frac{1}{2}r^2 - \frac{1}{2}r + k_1r} & \text{if } j=0,\\
    2^{k_2-k_1} q^{i+jk_2 - \frac{1}{2}r^2 -\frac{1}{2}r - \frac{1}{2}j^2 + \frac{1}{2}j - ij +k_1r} & \text{if } j = 1, \\ 
    q^{i+jk_2 - \frac{1}{2}r^2 -\frac{1}{2}r - \frac{1}{2}j^2 + \frac{1}{2}j - ij +k_1r} & \text{if } j \geq 2.
\end{cases}$$
Define
$$X(i,j,r) =\begin{cases} k_2 - k_1 + i - \frac{1}{2}r^2 - \frac{1}{2}r + k_1r & \text{if } j=0,\\
i+jk_2 - \frac{1}{2}r^2 -\frac{1}{2}r - \frac{1}{2}j^2 + \frac{1}{2}j - ij +k_1r & \text{if } j \geq 1.\end{cases}$$
The leading term  of
$$\sum_{r=1}^{k_1} \gamma(r)^{n-k_2}\sum_{i=0}^{k_1} \sum_{j=0}^{\min\{r, k_1-i\}} (-1)^{r-j}\binom{k_1}{i}f(i,j,r),$$
comes from the maximum of $X(i,j,r)$ over all allowed indices, which is
$$\max\big\{ X(i,j,r)\mid 1\leq r \leq k_1,\; 0\leq i \leq k_1, \; 0\leq j \leq \min\{r,k_1-i\}\big \} = X(0,k_1,k_1) = k_1k_2.$$
Therefore we get 
\begin{align*}
    \lim_{q \rightarrow \infty} \E[|\ker \psi_{\C_1, \C_2}|] &= 1 + \lim_{q \rightarrow \infty} q^{-n} \gamma^{n-k_2} (-1)^{k_1-k_1} \binom{k_1}{0} q^{k_1k_2}\\
    &= 1 + \lim_{q \rightarrow \infty} q^{k_1k_2 -n}, 
\end{align*}
which yields the desired result. 
\end{proof}

With the above result, we are able to show the first main result of this section. 

\begin{corollary} \label{cor:randomcodesinq}
    Let $\C_1, \C_2 \leq \F_q^n$ be uniformly random codes of dimensions $k_1$ and $k_2$, respectively, as in Definition~\ref{def:setuprand}. Then 
 $$\lim_{q \rightarrow \infty} \Pr(  \dim (\C_1\star\C_2) = \min\{k_1k_2, n\}) = 1.$$
\end{corollary}

\begin{proof}
Since $\im (\psi_{\C_1, \C_2}) = \C_1 \star \C_2$, we have that $|\ker \psi_{\C_1, \C_2}| < q^t$ if and only if $|\im \psi_{\C_1, \C_2}| > q^{k_1k_2 -t}$.
Hence $\Pr( \dim(\C_1 \star \C_2) > k_1k_2- t)  = \Pr(|\ker \psi_{\C_1, \C_2}| < q^t).$
Moreover, by Markov's inequality, we get $$\Pr(|\ker \psi_{\C_1, \C_2}| < q^t) \geq 1- \frac{\E\big[|\ker \psi_{\C_1, \C_2}|\big]}{q^t}.$$
We now consider two cases. If $k_1k_2 \leq n$, let $t=1$ and observe that by Theorem \ref{thm:qlimit}, 
\begin{align*}
   \lim_{q \rightarrow\infty} \Pr(|\ker \psi_{\C_1, \C_2}| < q) &\geq \lim_{q \rightarrow \infty}1 - \frac{2}{q} = 1.
\end{align*}
If $k_1k_2 > n$, let $t \coloneqq k_1k_2-n+1$ and observe that, again by Theorem \ref{thm:qlimit},  
\begin{align*}
   \lim_{q \rightarrow\infty} \Pr(|\ker \psi_{\C_1, \C_2}| < q^{k_1k_2-n+1}) &\geq \lim_{q \rightarrow \infty}1 - \frac{1 + q^{k_1k_2-n}}{q^{k_1k_2-n+1}} = 1.
\end{align*}
Putting everything together gives the desired statement.
\end{proof}

{ 
Beyond convergence in probability, the bounds derived above immediately yield the decay rate of the variance, showing it vanishes at least proportionally to the inverse of the field size, and at a faster inverse polynomial rate when $k_1k_2 < n$.

\begin{proposition} \label{prop:variance_q}
    Let $\C_1, \C_2 \leq \F_q^n$ be uniformly random codes of dimensions $k_1$ and $k_2$, respectively, as in Definition~\ref{def:setuprand}. Then
    $$
        \mathrm{Var}(\dim (\C_1\star \C_2)) = \begin{cases}
            \mathcal{O}_q(q^{k_1k_2-n}) & \text{if } k_1k_2 < n, \\
            \mathcal{O}_q(q^{-1}) & \text{if } k_1k_2 \geq n.
        \end{cases}
    $$
\end{proposition}
\begin{proof}
    Let $M=\min\{k_1k_2,n\}$ and define the dimension defect $X=M-\dim(\C_1 \star \C_2)$. Then $X$ is a non-negative integer bounded by $n$. We can bound the second moment as
    $$
        \mathbb{E}[X^2] = \mathbb{E}[X^2 \,|\, X\geq 1] \mathbb{P}(X\geq 1) \leq n^2 \mathbb{P}(X\geq 1).
    $$
    Since $M$ is constant, $\mathrm{Var}(\dim(\C_1\star \C_2)) = \mathrm{Var}(X)\leq \mathbb{E}[X^2] \leq n^2\mathbb{P}(X \geq 1)$. For $k_1k_2 < n$, the event $X \geq 1$ implies $|\ker \psi_{\C_1,\C_2}| > 1$, which is equivalent to $|\ker \psi_{\C_1,\C_2}| - 1 \geq 1$. Applying Markov's inequality and Theorem~\ref{thm:qlimit} yields $\mathbb{P}(X \geq 1) \leq \mathbb{E}[|\ker \psi_{\C_1,\C_2}| - 1] = \mathcal{O}_q(q^{k_1k_2-n})$. For $k_1k_2 \geq n$, the event $X \geq 1$ is equivalent to $|\ker \psi_{\C_1,\C_2}| \geq q^{k_1k_2-n+1}$. Again, applying Markov's inequality and Theorem~\ref{thm:qlimit} yields $\mathbb{P}(X \geq 1) \leq \mathbb{E}[|\ker \psi_{\C_1,\C_2}|] / q^{k_1k_2-n+1} = \mathcal{O}_q(q^{-1})$. As $n$ is fixed independent of $q$, the variance scales accordingly.
\end{proof}

Analogous to our experimental comparison for the expected dimension, we complement the asymptotic variance bounds with empirical data for small field sizes. Table~\ref{tab:distribution_experiments} in the Appendix presents the empirical distribution and sample variance of the star product dimension across the same parameter sets. Consistent with the preceding theorem, the sample variance decays rapidly as $q$ increases. Moreover, the sample frequencies highlight that over $\F_2$, there seems to remain a non-negligible probability of sampling code pairs that yield a relatively small star product dimension.
}

In the sequel, we show that when fixing the code $\C_1$ to be MDS and $\C_2$ to be uniformly random of a certain dimension relative to $\C_1$, the expected dimension of $\C_1 \star \C_2$ behaves as if~$\C_1$ was a uniformly random code. We start with a lemma that can be shown with a standard application of Möbius inversion.

\begin{Lemma} \label{lem:supp_subspaces}
Let $n,\ell\in \mathbb{N}$ with $\ell \leq n$ and $I\subseteq [n]$. Then the number of $\D \leq \F_q^n$ with $\dim \D = \ell$ and $\supp(\D) = I$ is 
\[\sum^{|I|}_{i=\ell}(-1)^{|I|-i} \binom{|I|}{i}\qbinom{i}{\ell}_q.\]
\end{Lemma}

\begin{proposition}\label{prop:expdim-mds-rnd}
Let $\C_1\leq \mathbb{F}_q^n$ be MDS with $\dim \C_1 = k_1$. Fix $k_2\leq n$ such that either $k_2=1$, or $k_2\geq n-k_1+1$. Let $\C_2 \leq \mathbb{F}_q^n$ be a uniformly random code with $\dim \C_2 = k_2$. Then$$ \mathbb{E}[\dim (\C_1\star\C_2)] =
\begin{cases}
    \frac{1}{q^n-1} \sum^n_{i=1}\binom{n}{i}(q-1)^{i}\min\{k_1,i\}&\text{if } k_2 = 1,\\
      \qbinom{n}{k_2}^{-1}_q \sum_{s=k_2}^n s\binom{n}{s}\left( \sum_{i=0}^{s-k_2} (-1)^i \qbinom{s-i}{k_2}_q\binom{s}{s - i}\right)&\text{if } k_2 \geq n-k_1+1.
\end{cases}$$
\end{proposition}

\begin{proof}
If $k_2=1$ then $\dim (\C_1\star \C_2) = \min \{\dim \C_1,|\mathrm{supp}(\C_2)|\}$. The result then follows by counting the number of $v\in \F_q^n$ with $|\mathrm{supp}(v)| = i$ for $i=1,\ldots,n$, up to a scalar multiple.
For the rest of the proof assume $k_2 \geq n-k_1+1$ and let $I = \supp(\C_2)$.

Suppose $|I|>k_1$. Then $\pi_{I}(\C_1)$ is an MDS code of dimension $k_1$ and $\pi_{I}(\C_2)$ is non-degenerate with $\dim(\pi_{I}(\C_2)) =k_2$. By Proposition \ref{prop:MDSnondeg} we have
\[ \dim(\C_1\star \C_2) = \dim ( \pi_{I}(\C_1)\star \pi_{I}(\C_2)) \geq \min \{|I|, \dim(\pi_{I}(\C_1))+\dim(\pi_{I}(\C_2)) - 1\} = |I|.\]
Furthermore, $\dim(\C_1\star \C_2) = \dim ( \pi_{I}(\C_1')\star \pi_{I}(\C_2')) \leq |I|$ as $\pi_{I}(\C_1),\pi_{I}(\C_2)\leq \F_q^{|I|}$, which proves $\dim(\C_1\star \C_2) =|\supp(\C_2)|$ for $|\supp(\C_2)|>k_1$.

Now suppose $|I|\leq k_1$. For $j\in I$ choose $c^j\in \C_1$ such that $c^j_j = 1$ and $c^j_i = 0$ for any $i\in I\setminus\{j\}$. This is possible as $\C_1$ is MDS with $\dim(\C_1)= k_1 \geq |I|$. Similarly, choose $d^j\in \C_2$ with $d^j_j = 1$. Then $\C_1\star \C_2 = \langle c^j\star d^j = e_j \; |\; j \in \supp(\C_2)\rangle$, so $\dim(\C_1\star \C_2) = |\supp(\C_2)|$.

Lastly, for $s\in \mathbb{N}$ let $\varphi(s)$ denote the number of subspaces $\C_2\leq \F_q^n$ of dimension $\ell$ and $|\supp(\C_2)|=s$. Then
\[
\sum_{\substack{\C_2 \leq \F_q^n \\ \dim \C_2 = \ell}} |\supp(\C_2)| = \sum^{n}_{s= \ell} s\varphi(s),
\]
and the result then follows by Lemma \ref{lem:supp_subspaces}.
\end{proof}

Note that Proposition \ref{prop:expdim-mds-rnd} does not cover the cases where $2\leq \dim \C_2 \leq n-\dim \C_1$. In the following example, we provide two monomially inequivalent MDS codes which yield different expected dimensions with respect to the star product for a value not covered by Proposition~\ref{prop:expdim-mds-rnd}. This also shows that the expected star product dimension is not an invariant under combinatorial equivalence, where by combinatorial equivalence we mean that the induced matroids of the linear codes are equivalent \cite{Oxley}.

\begin{example} \label{exmp:mdscodes}
Consider the two $[6,3]_7$ codes $\C_1,\C_2$ defined by the generator matrices
\[G_{\C_1} = \begin{bmatrix}
    1 & 0 & 0 & 4 & 5 & 2 \\
    0 & 1 & 0 & 6 & 1 & 1 \\
    0 & 0 & 1 & 5 & 6 & 5
\end{bmatrix}, \quad G_{\C_2} = \begin{bmatrix}
    1 & 0 & 0 & 1 & 1 & 6 \\
    0 & 1 & 0 & 4 & 1 & 4 \\
    0 & 0 & 1 & 6 & 2 & 4
\end{bmatrix}.\]
Both codes are MDS, so their underlying matroids are both uniform with the same parameters. If $\D \leq \F_7^6$ is any code chosen uniformly at random with dimension $\dim \D = 2$, then
\begin{align*}
    \mathbb{E}[\dim (\C_1 \star \D)] = \frac{13138498}{2288417},& \quad \mathbb{E}[\dim( \C_2 \star \D)] = \frac{13154050}{2288417}.
\end{align*}
However, for the case of $\dim \D = 3$, which is also not covered by Proposition \ref{prop:expdim-mds-rnd}, we have
\begin{align*}
    \mathbb{E}[\dim( \C_1 \star \D)] =   \mathbb{E}[\dim( \C_2 \star \D)] = \frac{72051027}{12044300}.
\end{align*}
\end{example}

The expected dimension under the star product when $\C_1$ is fixed is however an invariant under monomial equivalence, as shown by the proposition below. 
\begin{proposition}
    Let $\C_1,\C_2\leq \F_q^n$ be monomially equivalent and $\ell \in [n]$. Let $\D\leq \F_q^n$ be any code chosen uniformly at random with $\dim \D = \ell$. Then $\mathbb{E}[\dim ( \C_1 \star \D)] = \mathbb{E}[\dim (\C_2 \star \D)]$.
\end{proposition}
\begin{proof}
    
    Suppose $\C_1$ and $\C_2$ is monomially equivalent under $M$, so $\C_1 = \C_2M$. Then,
\[\sum_{\substack{\D \leq \F_q^n \\ \dim \D = \ell}} \dim (\C_1\star \D)=\sum_{\substack{\D \leq \F_q^n \\ \dim \D = \ell}} \dim (\C_2 M\star \D) =\sum_{\substack{\D \leq \F_q^n \\ \dim \D = \ell}} \dim (\C_2\star \D M) =\sum_{\substack{\D' \leq \F_q^n \\ \dim \D' = \ell}} \dim (\C_2 \star \D'),\]
as $M$ is a generalized permutation matrix.
\end{proof}

Lastly, we show that when $\C_1$ is a fixed MDS code and $\C_2$ is a uniformly random code whose dimension satisfies either $\dim \C_2 = 1$ or $\dim \C_2 \ge n - \dim \C_1 + 1$, the expected dimension of their star product behaves asymptotically in $q$ as it does when both $\C_1$ and $\C_2$ are uniformly random, as stated in Corollary~\ref{cor:randomcodesinq}.

\begin{proposition} \label{prop:mdsasymp}
 Let the setup be as in Proposition \ref{prop:expdim-mds-rnd}.
Then 
 $$\lim_{q \rightarrow \infty} \mathbb{P}(  \dim (\C_1\star\C_2) = \min\{k_1k_2, n\}) = 1.$$
\end{proposition}
\begin{proof}
We first show that $\lim_{q\rightarrow \infty}\mathbb{E}[\dim (\C_1\star\C_2)] = \min\{k_1k_2,n\}$. Consider first the case $k_2=1$. By Proposition \ref{prop:expdim-mds-rnd} then
\begin{align*}
    \lim_{q\rightarrow \infty} \mathbb{E}[\dim (\C_1\star\C_2)] &= \lim_{q \rightarrow \infty} 
    \frac{1}{q^n-1} \sum^n_{i=1}\binom{n}{i}(q-1)^{i}\min\{k_1,i\} \\
    &= \sum^{n}_{i=1} \binom{n}{i}\min\{k_1,i\} \lim_{q\rightarrow \infty} \frac{(q-1)^i}{q^n-1} \\
    &= \min\{k_1,n\} \lim_{q\rightarrow \infty} \frac{(q-1)^n}{q^n-1} \\
    &= \min \{k_1,n\} \lim_{q\rightarrow \infty}\frac{(1-\frac{1}{q})^n}{1-\frac{1}{q^n}}\\ 
    &= \min \{k_1,n\}.
\end{align*}
Similarly, for $k_2\geq n-k_1+1$ we have
\begin{align*}
    \lim_{q\rightarrow \infty} \mathbb{E}[\dim (\C_1\star\C_2)] &= \lim_{q\rightarrow \infty}  \qbinom{n}{k_2}^{-1}_q \sum_{s=k_2}^n s\binom{n}{s}\left( \sum_{i=0}^{s-k_2} (-1)^i \qbinom{s-i}{k_2}_q\binom{s}{s - i}\right) \\ 
    &= \sum^{n}_{s=k_2} \sum^{s-k_2}_{i=0} s \binom{n}{s} (-1)^i \binom{s}{s-i} \lim_{q \rightarrow \infty} \qbinom{s-i}{k_2}_q {\qbinom{n}{k_2}_q}^{-1} \\
    &= n \lim_{q\rightarrow \infty} \qbinom{n}{k_2}_q \qbinom{n}{k_2}_q^{-1} \\
    &= n  \\
    &= \min\{k_1k_2,n\}.
\end{align*}
Indeed, the last equality follows from the observation that the assumption 
$k_2 \ge n - k_1 + 1$ implies $k_1k_2 \ge n$, so $\min\{k_1k_2, n\} = n$.

Define the non-negative integer-valued random variable $X = \min\{k_1k_2,n\} - \dim (\C_1\star\C_2)$.
Then
\[
    \lim_{q\rightarrow \infty} \mathbb{E}[X] = \min \{k_1k_2,n\} - \lim_{q\rightarrow \infty}\mathbb{E}[\dim (\C_1\star\C_2)] = 0,
\]
so by Markov's inequality,
\[
    \mathbb{P}(\dim (\C_1\star\C_2) \neq \min\{k_1k_2,n\}) = \mathbb{P}(X \geq 1) \leq \mathbb{E}[X],
\]
and the result then follows by letting $q$ tend to $\infty$.
\end{proof}

\section{Asymptotic behavior as the code dimensions grow}
\label{sec:asympindim}

We now consider the asymptotic behavior of the dimension of the star product of two uniformly random codes, in the case where the dimensions $k_1, k_2$ of the respective codes $\C_1, \C_2$ are strictly increasing functions of the same variable $t$, while assuming $k_2$ does not grow too quickly compared to $k_1$. Before proving the main result of this section, we introduce some new notation and prove a few lemmas we will need.

\begin{notation}
Let $k_1,k_2 \in \mathbb{N}$ such that $ k_1\leq k_2$ and $\alpha\in \mathbb{R}$. Define
     \begin{align*}
      S^-_{k_1,k_2}(\alpha) &\coloneqq  \left\{A \in S^{k_1,k_2} \, \middle| \, 0 <  \rk(A) \leq \alpha k_1\right\}, \\
    S_{k_1,k_2}^+(\alpha) &\coloneqq \left\{A \in S^{k_1,k_2} \, \middle| \, \rk(A) > \alpha k_1\right\}.
    \end{align*}
The subscripts $k_1, k_2$ will be omitted when they are clear from context. For instance, we will simply write $S^-(\alpha)$ and $S^+(\alpha)$.
\end{notation}

\begin{Lemma}\label{lem:density}
Let $k_1,k_2\colon \mathbb{N}\rightarrow \mathbb{N}$ be strictly monotone increasing functions such that $k_1(t)\leq k_2(t)$ for all $t\in \mathbb{N}$.
Fix $\alpha\in(0,1)$. 
Then 
\begin{align*}
    \lim_{t\rightarrow \infty}\frac{|S^-(\alpha)|}{|S|} = 0 \quad \text{and} \quad \lim_{t\rightarrow \infty}\frac{|S^+(\alpha)|}{|S|} = 1.
\end{align*}
\end{Lemma}
\begin{proof}
We use
$$\qbinom{k_2(t)}{r}_q \le 
q^{r(k_2(t)-r+1)}$$ and bound $|S_r|$ by the number of $k_1(t)\times k_2(t)$ matrices of rank $r$ over $\mathbb{F}_q$, i.e. $|S_r| \leq q^{r(k_2(t)+k_1(t)) - r^2+r-k_1(t)}.$
Thus,
\begin{align*}
    |S^{-}(\alpha)| = \sum^{\lfloor \alpha k_1(t) \rfloor}_{r=1} |S_r|
    &\leq \sum^{\lfloor \alpha k_1(t) \rfloor}_{r=1} q^{r(k_2(t)+k_1(t))-r^2+r-k_1(t)}\\ 
    &\leq \lfloor \alpha k_1(t) \rfloor  q^{\lfloor \alpha k_1(t) \rfloor(k_2(t)+k_1(t))-(\lfloor \alpha k_1(t) \rfloor)^2+\lfloor \alpha k_1(t) \rfloor-k_1(t)}\\
    &\leq  \alpha k_1(t)  q^{\alpha k_1(t)(k_2(t)+k_1(t))-(\alpha k_1(t)-1)^2+\alpha k_1(t)-k_1(t)}\\
    &= q^{\alpha k_1(t) k_2(t) + (\alpha -\alpha^2)k_1(t)^2 + (3\alpha-1)k_1(t) + \log_q( \alpha k_1(t) ) }.
\end{align*}
Therefore,
\begin{align}
\frac{|S^-(\alpha)|}{|S|} &\leq \frac{q^{\alpha k_1(t) k_2(t) + (\alpha -\alpha^2)k_1(t)^2 + (3\alpha-1)k_1(t) + \log_q(\alpha k_1(t))}}{q^{k_1 k_2 -k_1}}  \nonumber \\
&= q^{(\alpha-1) k_1(t) k_2(t) + (\alpha -\alpha^2)k_1(t)^2 + 3\alpha k_1(t) + \log_q(\alpha k_1(t))} \rightarrow 0 \quad \text{as} \quad t\rightarrow \infty,\label{eq:Sminus2}
\end{align}
where \eqref{eq:Sminus2} follows as $\alpha<1$ and $k_1(t) \leq k_2(t)$ for all $t\in \mathbb{N}$. Since \smash{$\frac{|S^-(\alpha)|}{|S|} \geq 0$} as well, we get that \smash{$\lim_{t \rightarrow \infty} \frac{|S^-(\alpha)|}{|S|}=0$}. Finally, since  $|S^+(\alpha)| = |S|-1-|S^-(\alpha)|$, the latter limit of the lemma follows.
\end{proof}

\begin{Lemma}
Let $k_1,k_2\colon \mathbb{N}\rightarrow \mathbb{N}$ be strictly monotone increasing functions such that $k_1(t)\leq k_2(t)$ for all $t\in \mathbb{N}$.
Then for any $\alpha \in [0,1)$ we have
    $$\lim_{t \rightarrow \infty} \frac{\sum_{r=\lfloor\alpha k_1(t)\rfloor}^{k_1(t)} \sum_{\ell=0}^{k_2(t)-k_1(t)} \binom{k_2(t)-k_1(t)}{\ell}S_{r,\ell}q^{\ell}}{\sum_{r=\lfloor \alpha k_1(t) \rfloor}^{k_1(t)} \sum_{\ell=0}^{k_2(t)-k_1(t)} \binom{k_2(t)-k_1(t)}{\ell}S_{r,\ell}} = \lim_{t\rightarrow \infty} \exp \left( \frac{(q-1)k_2(t)}{q^{k_1(t)}+1}\right).
$$ \label{lem:weighteddensity}
\end{Lemma}

\begin{proof}
We begin by proving the claim for $\alpha = 0$, and then show that this implies it for any $0 \le \alpha <1$. Note that $\sum_{r=0}^{k_1(t)} S_{r,\ell} = q^{k_1(t)k_2(t) - k_1(t) - k_1(t)\ell}$. This allows us to rewrite the LHS of the statement as
    \begin{align*}
       \frac{q^{k_1(t) k_2(t) - k_1(t)}\sum_{\ell = 0}^{k_2(t)-k_1(t)} \binom{k_2(t)-k_1(t)}{\ell}q^{-\ell(k_1(t)-1)}}{q^{k_1(t)k_2(t)-k_1(t)}\sum_{\ell=0}^{k_2(t)-k_1(t)}\binom{k_2(t)-k_1(t)}{\ell}q^{-\ell k_1(t)}} &= \frac{(q^{-(k_1(t)-1)} +1)^{k_2(t)-k_1(t)}}{(q^{-k_1(t)} + 1)^{k_2(t)-k_1(t)}}\\
                &= \left(\frac{q^{k_1(t)} +q}{q^{k_1(t)}+1}\right)^{k_2(t)-k_1(t)}\\
                &= \left(1 + \frac{q-1}{q^{k_1(t)}+1}\right)^{k_2(t)-k_1(t)}, 
    \end{align*}
where the first equality follows from the Binomial Theorem.

We distinguish two cases according to the asymptotic behavior of $k_2(t)-k_1(t)$, which will be exhaustive by the assumed properties of $k_1$ and $k_2$. For the first case suppose $\lim_{t \rightarrow \infty} k_2(t) - k_1(t) =\infty$. It is well known that for functions $h,p\colon \mathbb{R}\rightarrow \mathbb{R}$ such that $\lim_{t \rightarrow \infty} h(t) = 0$ and $\lim_{t \rightarrow \infty} p(t) = \infty$, we have $$\lim_{t \rightarrow \infty}(1+h(t))^{p(t)} = \lim_{t \rightarrow \infty} \exp({h(t)p(t)}).$$
Hence, letting $h(t) \coloneqq  \frac{q-1}{q^{k_1(t)}+1}$ and $p(t) \coloneqq k_2(t)-k_1(t)$, 
we get the desired identity for $\alpha = 0$ and $\lim_{t \rightarrow \infty} k_2(t) - k_1(t) =\infty$.

For the second case, suppose $\lim_{t \rightarrow \infty} k_2(t) - k_1(t) = C < \infty.$ Since $\lim_{t \rightarrow \infty} k_1(t) = \infty$ the result follows as
\begin{align*}
    1= \lim_{t \rightarrow \infty} \left(1 + \frac{q-1}{q^{k_1(t)}+1}\right)^{k_2(t)-k_1(t)} &=  \lim_{t \rightarrow \infty} \exp\left({\frac{q-1}{q^{k_1(t)}+1}(k_2(t)-k_1(t))}\right). \\
    &=  \lim_{t \rightarrow \infty} \exp\left({\frac{(q-1)k_2(t)}{q^{k_1(t)}+1}}\right).
\end{align*}

Suppose now $0<\alpha <1$ and let $$N = \sum_{r=0}^{k_1(t)} \sum_{\ell=0}^{k_2(t)-k_1(t)} \binom{k_2(t)-k_1(t)}{\ell}S_{r,\ell}q^{\ell},$$ and $$D =\sum_{r=0}^{k_1(t)} \sum_{\ell=0}^{k_2(t)-k_1(t)} \binom{k_2(t)-k_1(t)}{\ell}S_{r,\ell} = |S|.$$ Additionally, partition the sums of $D$ and $N$ as follows:
\begin{align*}
    N' &= \sum_{r=\lfloor\alpha k_1(t)\rfloor}^{k_1(t)} \sum_{\ell=0}^{k_2(t)-k_1(t)} \binom{k_2(t)-k_1(t)}{\ell}S_{r,\ell}q^{\ell},\\
    \Delta N &= \sum_{r=0}^{\lfloor \alpha k_1(t) \rfloor -1} \sum_{\ell=0}^{k_2(t)-k_1(t)} \binom{k_2(t)-k_1(t)}{\ell}S_{r,\ell}q^{\ell}, \\
    D' &= \sum_{r=\lfloor \alpha k_1(t) \rfloor}^{k_1(t)} \sum_{\ell=0}^{k_2(t)-k_1(t)} \binom{k_2(t)-k_1(t)}{\ell}S_{r,\ell},\\
    \Delta D &= \sum_{r=0}^{\lfloor \alpha k_1(t) \rfloor -1} \sum_{\ell=0}^{k_2(t)-k_1(t)} \binom{k_2(t)-k_1(t)}{\ell}S_{r,\ell}. \\
\end{align*}
Note that when $r=0$, we have $\sum_{\ell=0}^{k_2(t)-k_1(t)} \binom{k_2(t)-k_1(t)}{\ell}S_{0,\ell} =1$. Hence, when considering  $\frac{\Delta D -1}{D}$, we can apply Lemma \ref{lem:density} and its proof to obtain  $\lim_{t \rightarrow \infty} \frac{\Delta D}{D} = \lim_{t \rightarrow \infty} \frac{\Delta D - 1}{D} =  0$ as 
\[
    \frac{\Delta D -1}{D} \leq q^{(\alpha-1) k_1(t) k_2(t) + (\alpha-\alpha^2) k_1(t)^2 + \alpha k_1(t) + \log_q(\alpha k_1(t))}.
\]
Furthermore, as $1\leq q^\ell \leq q^{k_2(t)-k_1(t)}$ for all $\ell \leq k_2(t)-k_1(t)$, we have
\begin{align*}
    \frac{\Delta N - 1}{N} &\leq q^{k_2(t)-k_1(t)} \frac{\Delta D -1}{D} \\
    &\leq q^{(\alpha-1) k_1(t) k_2(t) + (\alpha-\alpha^2) k_1(t)^2 + k_2(t) + (\alpha -1)k_1(t) + \log_q(\alpha k_1(t))}.
\end{align*}
This allows us to conclude that $\lim_{t\rightarrow \infty} \frac{\Delta N}{N} = \lim_{t\rightarrow \infty} \frac{\Delta N -1 }{N} = 0$. The result then follows as
\begin{align*}
   \lim_{t\rightarrow \infty} \frac{N'}{D'} &= \lim_{t\rightarrow \infty}\frac{N-\Delta N}{D- \Delta D}\\
   &= \lim_{t\rightarrow \infty} \frac{N}{D} \left(\frac{1- \frac{\Delta N}{N}}{1- \frac{\Delta D}{D}} \right)  \\
   &= \lim_{t\rightarrow \infty} \frac{N}{D}\\ 
     &= \lim_{t\rightarrow \infty} \exp \left( \frac{(q-1)k_2(t)}{q^{k_1(t)}+1}\right)
  \end{align*}
where the last equality follows from the case $\alpha = 0$. 
\end{proof}

We will now consider the asymptotics with respect to the dimension of the codes $\C_1$ and $\C_2$. Specifically, if $\dim \C_1 = k_1$ and $\dim \C_2 = k_2$ then we will sample the codes from the ambient space $\F_q^{k_1 k_2}$, as we then expect asymptotically almost all such pairs of codes to have a star product equal to the entire ambient space. This is motivated by the fact that when taking the star product of two random codes then a generating set of $\C_1 \star \C_2$ will have size
\[ \dim\C_1\dim \C_2 - \frac{\dim(\C_1\cap \C_2) (\dim(\C_1\cap \C_2)-1)}{2},\]
where the expected intersection will tend to 0, as proven in the following result.

\begin{Lemma} \label{lemma:expectedintersection}
Fix $n,k_1,k_2\in \mathbb{N}$ with $k_1\leq k_2\leq n$. Let $\C_1,\C_2\leq \F_q^n$ be uniformly random codes of dimension $k_1$ and $k_2$, respectively. Then
\begin{align*}
\mathbb{E}[\dim(\C_1 \cap \C_2)] = \left( {{\qbinom{n}{k_1}_q\qbinom{n}{k_2}_q}}\right)^{-1}{{\sum\limits^{k_1}_{i = 1} i\qbinom{n}{i}_q \qbinom{n-i}{k_1-i}_q q^{(k_1-i)(k_2-i)} \qbinom{n-k_1}{k_2-i}_q}}. 
\end{align*}
\end{Lemma}
\begin{proof}
Fix a subspace $W\leq \F_q^n$ and $U'\leq \F_q^n/W$ with $\dim (U') = k_1-i$. We then count the number of $V' \leq \F_q^n/W$ with $\dim(V') = k_2-i$ such that $U'\cap V' = \{0\}$. There are $\prod^{k_2-i-1}_{j=0}(q^{n-i}-q^{k_1-i+j})$ ordered bases spanning a subspace of dimension $k_2-i$ which intersects~$U'$ trivially, which we divide by the number of different bases for the same subspace:
\begin{align}
\frac{(q^{n-i}-q^{k_1-i})(q^{n-i}-q^{k_1-i+1})\ldots(q^{n-i}-q^{k_1-i+k_2-i-1})}{(q^{k_2-i}-1)(q^{k_2-i}-q)\ldots(q^{k_2-i}-q^{k_2-i-1})} = q^{(k_1-i)(k_2-i)}\qbinom{n-k_1}{k_2-i}_q. \label{eq:eDim}
\end{align}
Lastly, we multiply \eqref{eq:eDim} by the number of ways to choose $U'$ given a fixed $W$, and the number of ways to choose $W$. The result then follows as
\begin{equation*}
    |(U,V)\leq \F_q^n \times \F_q^n \colon U\cap V = W|= |(U',V')\leq (\F_q^n /W) \times (\F_q^n/ W) \colon U'\cap V' = \{0\}|.\qedhere
\end{equation*} 
\end{proof}

\begin{proposition}\label{prop:limitinter}
Let $k_1,k_2\colon \mathbb{N}\rightarrow \mathbb{N}$ be strictly monotone increasing functions such that $k_1(t)\leq k_2(t)$ for all $t\in \mathbb{N}$. Let $n(t)=k_1(t)k_2(t)$. Let $\C_1, \C_2 \leq \F_q^{n(t)}$ be uniformly random codes of dimensions $k_1(t)$ and $k_2(t)$, respectively. Then $$\lim_{t\rightarrow \infty} \E[\dim (\C_1\cap \C_2)] = 0.$$
\end{proposition}\begin{proof}
From \cite{ihringer2015} we have the bounds
\[
    q^{k(n-k)}\leq \qbinom{n}{k}_q \leq \frac{7}{2}q^{k(n-k)},
\]
from which we obtain
\begin{align}
q^{k_1(t)(n(t)-k_1(t))+k_2(t)(n(t)-k_2(t))} \leq \qbinom{n(t)}{k_1(t)}_q \qbinom{n(t)}{k_2(t)}_q. \label{eq:D}
\end{align}
For $1\leq i \leq k_1(t)$ we have
\begin{align*}
    &\qbinom{n(t)}{i}_q \qbinom{n(t)-i}{k_1(t)-i}_q q^{(k_1(t)-i)(k_2(t)-i)}\qbinom{n(t)-k_1(t)}{k_2(t)-i}_q \\
    \leq&\; \left( \frac{7}{2}\right)^3q^{i(k_1(t)+k_2(t)-2k_1(t)k_2(t)) + k_1(t)(n(t)-k_1(t)) + k_2(t)(n(t)-k_2(t))},
\end{align*}
so we obtain
\begin{align}
    &\sum^{k_1(t)}_{i=1} i \qbinom{n(t)}{i}_q \qbinom{n(t)-i}{k_1(t)-i}_q q^{(k_1(t)-i)(k_2(t)-i)}\qbinom{n(t)-k_1(t)}{k_2(t)-i}_q \notag \\
    \leq&\; \sum^{k_1(t)}_{i=1} k_1(t) \qbinom{n(t)}{i}_q \qbinom{n(t)-i}{k_1(t)-i}_q q^{(k_1(t)-i)(k_2(t)-i)}\qbinom{n(t)-k_1(t)}{k_2(t)-i}_q \notag\\
    \leq&\; k_1(t)(k_1(t)+1)\left(\frac{7}{2}\right)^3q^{k_1(t)+k_2(t)-2n(t) + k_1(t)(n(t)-k_1(t)) + k_2(t)(n(t)-k_2(t))}. \label{eq:N}
\end{align}
Combining \eqref{eq:D} with \eqref{eq:N} and Lemma \ref{lemma:expectedintersection}, we have
\begin{align*}
    \E[\dim (\C_1 \cap \C_2)] \leq q^{-2k_1(t)k_2(t) + k_1(t) + k_2(t) + \log_q\left( k_1(t)(k_1(t)+1)\frac{343}{8}\right)},
\end{align*}
and the result follows as $k_1(t) \leq k_2(t)$ for all $t$.
\end{proof}

We now prove that the kernel of $\psi_{\C_1,\C_2}$ is particularly well-behaved when the dimension of $\C_2$ does not grow too quickly compared to that of $\C_1$. We specify this in the following definition.

\begin{definition}
    We say that functions $k_1,k_2\colon \mathbb{N}\rightarrow \mathbb{N}$ are \emph{admissible} if they are strictly monotone increasing such that $k_1(t)\leq k_2(t)$ for all $t\in \mathbb{N}$, and for all $\alpha \in (0,2-\log_q(2q-1))$, then
    $$\lim_{t\rightarrow \infty} \frac{k_1(t)k_2(t)}{q^{\alpha k_1(t)}} = 0.$$
\end{definition}

This assumption reflects a technical limitation of our approach rather than an intrinsic necessity, and we do not claim it to be optimal. The proof of Theorem~\ref{thm:asympkerneldimension2} relies on it as the number of zeros of a non-zero bilinear form in $\mathcal{B}$ is at most $(2q-1)^{k_2-k_1 - 2}$. Note that for $q=2$ we have $\alpha < 0.415$, and that $\lim_{q\rightarrow \infty} 2-\log_q(2q-1) = 1.$ Note moreover that 
strictly monotone and increasing functions $k_1,k_2\colon \mathbb{N}\rightarrow \mathbb{N}$ such that $k_1(t)\le k_2(t)$ and that both grow polynomially in $t$ are admissible for all $q$.

\begin{theorem}\label{thm:asympkerneldimension2}
Let $k_1,k_2\colon \mathbb{N}\rightarrow \mathbb{N}$ be admissible. Let $n(t)=k_1(t)k_2(t)$. Let $\C_1, \C_2 \leq \F_q^{n(t)}$ be uniformly random codes of dimensions $k_1(t)$ and $k_2(t)$, respectively, as in Definition~\ref{def:setuprand}. Then 
\smash{$\lim_{t\rightarrow \infty} \E [|\ker \psi_{\C_1,\C_2}|] = 2.$}
\end{theorem}\begin{proof}
The expected kernel size can also be written as
\begin{align*}
    &\E[ |\ker {\psi_{\C_1,\C_2}}|] =\\
    &\:\sum_{r=0}^{k_1(t)}\sum_{\ell = 0}^{k_2(t)-k_1(t)} \binom{k_2(t)-k_1(t)}{\ell} S_{r,\ell} \left(  \frac{Z(B)}{q^{k_1(t)+k_2(t)}}\right)^{n(t)-k_2(t)} q^{\ell - (k_2(t)-k_1(t))},
\end{align*}
where $Z(B)$ is the number of zeroes of a bilinear form in $S_{r,I}$, where $I\subseteq [k_2]\setminus [k_1]$ is any subset of size $|I|=\ell,$ and $S_{r,\ell} = |S_{r,V}|$.

Let $\alpha\in(0,2-\log_q(2q-1))$. By Lemma \ref{lem:density} and its proof we have that, for $t\rightarrow \infty$,
\begin{align}
\frac{|S^-(\alpha)|}{|S|} &\leq \frac{q^{\alpha k_1(t) k_2(t) + (\alpha -\alpha^2)k_1(t)^2 + (\alpha-1)k_1(t) + \log_q(\alpha k_1(t))}}{q^{k_1 k_2 -k_1}}  \nonumber \\
&= q^{(\alpha-1) k_1(t) k_2(t) + (\alpha -\alpha^2)k_1(t)^2 + \alpha k_1(t) + \log_q(\alpha k_1(t))} \rightarrow 0, \label{eq:Sminus3}
\end{align}
and $\lim_{t \rightarrow \infty} \frac{|S^+(\alpha)|}{|S|} = 1$. Let $B\in S^+(\alpha)$, so $r\geq \alpha k_1(t)$. Then
$$Z(B) \geq (q^{k_1(t)}+q-1)q^{k_1(t)+k_2(t)-k_1(t)-1} \geq  q^{k_1(t)+k_2(t)-1}\left(1 + \frac{q-1}{q^{k_1(t)}}\right),$$
which yields
\begin{align*}
    \left( \frac{Z(B)}{q^{k_1(t)+k_2(t)}} \right)^{n(t)-k_2(t)} q^{\ell -(k_2(t)-k_1(t))} &\geq \frac{q^{-(k_2(t)-k_1(t))}}{q^{n(t)-k_2(t)}}\left(1 + \frac{q-1}{q^{k_1(t)}}\right)^{n(t)-k_2(t)} \\ 
    &= \frac{1}{q^{n(t)-k_1(t)}}\left(1 + \frac{q-1}{q^{k_1(t)}}\right)^{n(t)-k_2(t)}.
\end{align*}
Since $n(t)=k_1(t)k_2(t)$, we obtain
\[
\sum_{B\in S^+(\alpha)} \left( \frac{Z(B)}{q^{k_1(t)+k_2(t)}} \right)^{k_2(t)(k_1(t)-1)} q^{\ell -(k_2(t)-k_1(t))} \geq \frac{|S^+(\alpha)|}{|S|}\left(1 + \frac{q-1}{q^{k_1(t)}}\right)^{k_2(t)(k_1(t)-1)}.
\]
Then
\begin{align*}
    \lim_{t\rightarrow \infty}  \left(1 + \frac{q-1}{q^{k_1(t)}}\right)^{k_2(t)(k_1(t)-1)} 
    =\; \lim_{t\rightarrow \infty} \exp \left( \frac{q-1}{q^{k_1(t)}} (k_2(t)(k_1(t)-1)) \right)
    = 1,
\end{align*}
so
\begin{align} \label{eq:upperboundthis}
 \lim_{t\rightarrow \infty} \sum_{B\in S^+(\alpha)} \left( \frac{Z(B)}{q^{k_1(t)+k_2(t)}} \right)^{k_1(t)k_2(t)-k_2(t)} q^{\ell -(k_2(t)-k_1(t))} \geq 1.
\end{align}

Thus, we need to show that \eqref{eq:upperboundthis} is also bounded from above by 1. Note that
\[
    \frac{Z(B)}{q^{k_1(t)+k_2(t)}} = \frac{(q^{r}+q-1)q^{k_1(t)+k_2(t)-r-1}}{q^{k_1(t)+k_2(t)}} = q^{-1}\left(1+\frac{q-1}{q^r}\right)\leq q^{-1}\left(1+\frac{q-1}{q^{\alpha k_1(t)}} \right),
\]
which yields
\[
    \left(\frac{Z(B)}{q^{k_1(t)+k_2(t)}}\right)^{n(t)-k_2(t)} = q^{-(n(t)-k_2(t))}\left(1+\frac{q-1}{q^r}\right)^{n(t)-k_2(t)}
\]
and
\begin{align*}
    \lim_{t\rightarrow \infty}  \left(1+\frac{q-1}{q^r}\right)^{n(t)-k_2(t)} &= \lim_{t\rightarrow \infty} \exp\left((n(t)-k_2(t))\ln \left(1+\frac{q-1}{q^r}\right) \right) \\
    &\leq \lim_{t\rightarrow \infty} \exp\left((n(t)-k_2(t))\frac{q-1}{q^r}\right) \\
    &\leq  \lim_{t\rightarrow \infty}\exp\left(\frac{q-1}{q^{\alpha k_1(t)}}k_2(t)(k_1(t)-1) \right) \\
    &=  \lim_{t\rightarrow \infty}\exp\left((q-1)\frac{k_1(t)k_2(t)}{q^{\alpha k_1(t)}} \right)\\
    &= 1.
\end{align*}
Thus, for sufficiently large $t$, we have 
\begin{align*}
    \left(\frac{Z(B)}{q^{k_1(t)+k_2(t)}}\right)^{n(t)-k_2(t)} \leq q^{-(n(t)-k_2(t))}\exp\left((q-1)\frac{k_1(t)k_2(t)}{q^{\alpha k_1(t)}} \right),
\end{align*}
which in turn yields
\begin{align}
    &\sum_{r=\alpha k_1(t)}^{k_1(t)}\sum_{\ell = 0}^{k_2(t)-k_1(t)} \binom{k_2(t)-k_1(t)}{\ell} S_{r,\ell} \left(  \frac{Z(B)}{q^{k_1(t)+k_2(t)}}\right)^{n(t)-k_2(t)} q^{\ell - (k_2(t)-k_1(t))} \notag\\
    \leq& \: q^{-(n(t)-k_1(t))}\exp\left((q-1)\frac{k_1(t)k_2(t)}{q^{\alpha k_1(t)}} \right)\sum_{r=\alpha k_1(t)}^{k_1(t)}\sum_{\ell = 0}^{k_2(t)-k_1(t)} \binom{k_2(t)-k_1(t)}{\ell} S_{r,\ell}q^{\ell}. \label{eq:upperbound}
\end{align}
Taking the limit of \eqref{eq:upperbound} it then follows by Lemma \ref{lem:weighteddensity} that
\begin{align*}
&\lim_{t\rightarrow \infty } q^{-(n(t)-k_1(t))}\exp\left((q-1)\frac{k_1(t)k_2(t)}{q^{\alpha k_1(t)}} \right)\sum_{r=\alpha k_1(t)}^{k_1(t)}\sum_{\ell = 0}^{k_2(t)-k_1(t)} \binom{k_2(t)-k_1(t)}{\ell} S_{r,\ell}q^{\ell} \\ 
=& \lim_{t\rightarrow \infty } \left(q^{-(n(t)-k_1(t))}\right) \lim_{t\rightarrow \infty} \left(\sum_{r=\alpha k_1(t)}^{k_1(t)}\sum_{\ell = 0}^{k_2(t)-k_1(t)} \binom{k_2(t)-k_1(t)}{\ell} S_{r,\ell}q^{\ell}\right) \\ 
=& \lim_{t\rightarrow \infty } \left(\frac{1}{|S|}\right) \lim_{t\rightarrow \infty}\left(\sum_{r=\alpha k_1(t)}^{k_1(t)}\sum_{\ell = 0}^{k_2(t)-k_1(t)} \binom{k_2(t)-k_1(t)}{\ell} S_{r,\ell}\right)\\
=&\lim_{t\rightarrow \infty} \frac{|S^+(\alpha)|}{|S|} =1.
\end{align*}
We conclude
\[\lim_{t\rightarrow \infty}\sum_{B\in S^+(\alpha)} \left( \frac{Z(B)}{q^{k_1(t)+k_2(t)}} \right)^{k_1(t)k_2(t)-k_2(t)}q^{\ell - (k_2(t)-k_1(t))}=1.
\]

Let $B\in S^-(\alpha)$. By Lemma \ref{Lemma:zerosbyrank}, $Z(B) \leq (2q-1)q^{k_1(t)+k_2(t)-2}$ and since we have $q^{\ell-(k_2(t)-k_1(t))}\leq 1$ for all $\ell \leq k_2(t)-k_1(t)$, we have
\begin{align*}
    \sum_{B\in S^-(\alpha)}\left(\frac{Z(B)}{q^{k_1(t)+k_2(t)}}\right)^{n(t)-k_2(t)}q^{\ell-(k_2(t)-k_1(t))} \leq \left(\frac{2q-1}{q^2}\right)^{n(t)-k_2(t)} |S^-(\alpha)|.
\end{align*}
Substituting $n(t)=k_1(t)k_2(t)$ it then follows by (\ref{eq:Sminus3}) that
\begin{align*}
    &\left(\frac{2q-1}{q^2}\right)^{k_2(t)(k_1(t)-1)} |S^-(\alpha)| \\
    \leq&\; q^{(\alpha - 2 + \log_q(2q-1))k_1(t)k_2(t) + (\alpha-\alpha^2)k_1(t)^2 + (2-\log_q(2q-1))k_2(t) + (\alpha-1)k_1(t) + \log_q(\alpha k_1(t))}.
\end{align*}
Since $\alpha<2-\log_q(2q-1)$ we obtain
\[
\lim_{t\rightarrow \infty}\sum_{B\in S^-(\alpha)}\left(\frac{Z(B)}{q^{k_1(t)+k_2(t)}}\right)^{k_1(t) k_2(t)-k_2(t)}q^{\ell-(k_2(t)-k_1(t))} =0,
\]
which concludes the proof.
\end{proof}

\begin{remark}
Note that by Theorem~\ref{thm:asympkerneldimension2} we have that $\lim_{t\rightarrow \infty} \mathbb{E}\big[|\ker \psi_{\C_1,\C_2}|\big]=\infty$ when \smash{$\lim_{t\rightarrow \infty} \frac{k_1(t)k_2(t)}{q^{k_1(t)}} = \infty$}. This follows directly from the proof when determining the lower bound. Furthermore, the proof and numerical experiments suggest the conjecture that $$\lim_{t\rightarrow \infty} \mathbb{E}\big[|\ker \psi_{\C_1,\C_2}|\big]= \exp\left(\frac{q-1}{q^{k_1(t)}}(k_2(t))(k_1(t)-1)\right)+1,$$ independently of the value of $\lim_{t\rightarrow \infty} \frac{k_1(t)k_2(t)}{q^{k_1(t)}}$.
\end{remark}

Having established the concentration properties of the kernel size and the growth behavior of the relevant combinatorial quantities, we are now ready to study the asymptotics of the expected star-product dimension as both code dimensions increase. The following results formalizes this by showing that, under the assumed admissible growth conditions, the expected kernel size stabilizes and the star product almost surely attains full dimension.
\\

\begin{theorem}\label{thm:expectedink}Let $k_1,k_2\colon \mathbb{N}\rightarrow \mathbb{N}$ be admissible. Let $n\colon \mathbb{N}\rightarrow \mathbb{N}$ such that $n(t) \geq k_1(t)k_2(t)$ for all $t\in \mathbb{N}$. Let $\C_1, \C_2 \leq \F_q^{n(t)}$ be uniformly random codes of dimensions $k_1(t)$ and $k_2(t)$, respectively, as in Definition~\ref{def:setuprand}. Then, for sufficiently large $t$, 
\[
\mathbb{P}\left( \dim (\C_1\star\C_2) = k_1(t) k_2(t) \right) \geq 1 - \left(\frac{2q-1}{q^2}\right)^{n(t)-k_1(t)k_2(t)}.
\]
\end{theorem}
\begin{proof}
For any codes $\C_1, \C_2 \leq \F_q^{n(t)}$ as in the statement, let $\C_1', \C_2' \leq \F_q^{k_1(t)k_2(t)}$ be the corresponding codes obtained by puncturing $\C_1$ and $\C_2$ in the last $n(t) - k_1(t)k_2(t)$ coordinates. 

Let $\mathcal{N}$ denote the event that $\dim(\C_1 \star \C_2) = k_1(t)k_2(t)$, and for any $j \in \mathbb{N}$ let $\mathcal{E}_j$ be the event that $|\ker\psi_{\C_1', \C_2'}| = j$. 
Note that $\dim(\C_1 \star \C_2) = k_1(t)k_2(t)$ if and only if $|\ker\psi_{\C_1, \C_2}| = 1$, which occurs if and only if, for all non-zero $\varphi \in \ker\psi_{\C_1', \C_2'}$, there exists $i \in \{k_1(t)k_2(t) + 1, \ldots, n(t)\}$ such that $\varphi(G_1^{(i)}, G_2^{(i)}) \neq 0$.

To study the event $\mathcal{E}_j$, let 
$\ker\psi_{\C_1', \C_2'} \setminus \{0\} = \{\varphi_1, \ldots, \varphi_{j-1}\}$. 
Then
\begin{align}
    \mathbb{P} \left(\overline{\mathcal{N}} \mid \mathcal{E}_j\right)
    &= \mathbb{P} \left(
        \bigcup_{i=1}^{j-1}
        \left\{
            \varphi_i(G_1^{(k_1(t)k_2(t)+1)}, G_2^{(k_1(t)k_2(t)+1)}) 
            = \cdots = 
            \varphi_i(G_1^{(n(t))}, G_2^{(n(t))}) = 0
        \right\}
    \right)\notag \\
    &\leq \sum_{i=1}^{j-1} 
        \mathbb{P} \left(\varphi_i(x, y) = 0\right)^{n(t) - k_1(t)k_2(t)} \notag \\
    &\leq \sum_{i=1}^{j-1} 
        \left(\frac{2q - 1}{q^2}\right)^{n(t) - k_1(t)k_2(t)} \label{eq:notevent} \\
    &= (j - 1)
        \left(\frac{2q - 1}{q^2}\right)^{n(t) - k_1(t)k_2(t)}, \notag
\end{align}
where $x \in \F_q^{k_1(t)}$ and $y \in \F_q^{k_2(t)}$ are drawn independently and uniformly at random. 
Note that~\eqref{eq:notevent} follows from Lemma~\ref{Lemma:zerosbyrank}, since each $\varphi_i$ is nonzero. By the law of total probability, 
\begin{align*}
    \mathbb{P}(\overline{\mathcal{N}}) 
    &= \sum_{j \in \mathbb{N}} 
        \mathbb{P}(\mathcal{E}_j)\,
        \mathbb{P}(\overline{\mathcal{N}} \mid \mathcal{E}_j) \\
    &\leq 
        \left(\frac{2q - 1}{q^2}\right)^{n(t) - k_1(t)k_2(t)} 
        \left( \sum_{j \in \mathbb{N}} 
            \mathbb{P}(\mathcal{E}_j)(j - 1)
        \right) \\ 
    &= 
        \left(\frac{2q - 1}{q^2}\right)^{n(t) - k_1(t)k_2(t)} 
        \left( \mathbb{E} \big[|\ker\psi_{\C_1', \C_2'}|\big] - 1 \right).
\end{align*}
The result then follows from Theorem~\ref{thm:asympkerneldimension2}.
\end{proof}

\begin{corollary}
Let $k_1,k_2\colon \mathbb{N}\rightarrow \mathbb{N}$ be admissible. Let $\C_1, \C_2 \leq \F_q^{k_1(t)k_2(t)}$ be uniformly random codes of dimensions $k_1(t)$ and $k_2(t)$, respectively, as in Definition~\ref{def:setuprand}. For any $\varepsilon>0$ there exists $t_{\varepsilon}\in \mathbb{N}$ such that, for all $t\geq t_{\varepsilon}$, for every non-negative integer $\ell$ we have
\[
    \mathbb{P}\left(\dim (\C_1\star\C_2) \geq k_1(t)k_2(t) -\ell \right) \geq 1 - \frac{2+\varepsilon}{q^{\ell+1}}.
\]
\end{corollary}
\begin{proof}
By Theorem~\ref{thm:asympkerneldimension2}, there exists $t_{\varepsilon} \in \mathbb{N}$ such that, for all $t \geq t_{\varepsilon}$, $
    \mathbb{E}\big[|\ker\psi_{\C_1, \C_2}|\big] \leq 2 + \varepsilon.$
By Markov's inequality, for $t \geq t_{\varepsilon}$ we have
\[
    \mathbb{P}\left(|\ker\psi_{\C_1, \C_2}| < \delta \right) 
    \geq 1 - \frac{\mathbb{E}\big[|\ker\psi_{\C_1, \C_2}|\big]}{\delta}
    \geq 1 - \frac{2 + \varepsilon}{\delta}.
\]
The result follows by taking $\delta = q^{\ell + 1}$, since
\[
    \mathbb{P}\left(|\ker\psi_{\C_1, \C_2}| < q^{\ell + 1}\right)
    = \mathbb{P} \left(\dim(\C_1 \star \C_2) \geq k_1(t)k_2(t) - \ell\right). 
\]
\end{proof}

Theorem~\ref{thm:expectedink} addresses the regime in which the ambient length is no smaller than the product of the code dimensions. By applying puncturing arguments, one obtains an analogous result when the ambient space is smaller.

\begin{theorem} \label{thm:expectedink2}
Let $k_1,k_2\colon \mathbb{N}\rightarrow \mathbb{N}$ be admissible. Let $n\colon \mathbb{N}\rightarrow \mathbb{N}$ such that $n(t) < k_1(t)k_2(t)$ for all $t\in \mathbb{N}$. Let $\C_1, \C_2 \leq \F_q^{n(t)}$ be uniformly random codes of dimensions $k_1(t)$ and $k_2(t)$, respectively, as in Definition~\ref{def:setuprand}. Then, for sufficiently large $t$, 
\[
\mathbb{P}\left( \dim (\C_1\star\C_2) = n(t) \right) \geq 1 - \left(\frac{2q-1}{q^2}\right)^{t}.
\]
\end{theorem}
\begin{proof}
For $i=1,2$, let $G_i$ be the systematic generator matrix of $\C_i\leq \F_q^{n(t)}$. Set $m(t)=k_1(t)k_2(t)$ and extend $G_i$ by concatenating $m(t)-n(t)+t$ additional columns sampled uniformly and independently from $\F_q^{k_i(t)}$. Let \smash{$\tilde{G}_i\in \F_q^{k_i(t)\times (m(t)+t)}$} be the resulting matrix and \smash{$\tilde{\C}_i= \rowsp (\tilde{G}_i)$}.

If $\pi \colon \F_q^{m(t)+t}\rightarrow \F_q^{n(t)}$ denotes the projection onto the first $n(t)$ coordinates, then $\pi( \tilde{\C}_i) = \C_i$. Let $\mathcal{N}$ be the event that $\dim (\tilde{\C}_1 \star \tilde{\C}_2)=k_1(t)k_2(t)$. By Theorem~\ref{thm:expectedink} for sufficiently large $t$ we have
\[
    \mathbb{P}(\mathcal{N}) \geq 1 - \left( \frac{2q-1}{q^2}\right)^{t}.
\]
Let $\mathcal{M}$ be the event $\dim (\C_1\star\C_2) = n(t)$. We then show that $\overline{\mathcal{M}}$ implies $\overline{\mathcal{N}}$. Suppose that $\dim (\C_1\star\C_2) < n(t)$, so there exists a non-zero $v\in (\C_1\star \C_2)^\perp$. Consider $\tilde{v} =(v,0,\ldots,0)\in \F_q^{m(t)+t}$. Since $\tilde{\C}_1\star \tilde{\C}_2$ is spanned by $\tilde{c}_1\star \tilde{c}_2$, where $\tilde{c}_1$ and $\tilde{c}_2$ are basis elements of $\tilde{\C}_1$ and $\tilde{\C}_2$, respectively, we have
$
    \langle \tilde{v},\tilde{c}_1\star \tilde{c}_2\rangle = \langle v,\pi(\tilde{c}_1)\star \pi(\tilde{c}_2)\rangle =0.
$
Thus, $\tilde{v}\in (\tilde{\C}_1\star \tilde{\C}_2)^\perp$ and $\mathbb{P}(\mathcal{M}) \geq \mathbb{P}(\mathcal{N})$.
\end{proof}

We then conveniently summarize the conclusions of Theorem~\ref{thm:expectedink} and Theorem~\ref{thm:expectedink2} in a single result.
\begin{corollary} \label{cor:mainresult}Let $k_1,k_2\colon \mathbb{N}\rightarrow \mathbb{N}$ be admissible.  Let $\C_1, \C_2 \leq \F_q^{n(t)}$ be uniformly random codes of dimensions $k_1(t)$ and $k_2(t)$, respectively, as in Definition~\ref{def:setuprand}. Then, for sufficiently large $t$, 
\[
\mathbb{E}[\dim (\C_1\star\C_2)] = \min \{k_1(t)k_2(t),n(t)\} + o_t(1).
\]
\end{corollary}

{ 
\begin{remark}\label{remark:otherpaper}
    Although our kernel estimation approach was developed independently of the rank-1 matrix analysis in \cite{7282444}, there is a partial overlap in the conclusions regarding the expected dimension of star products under specific parameter configurations. We briefly compare our asymptotic bounds on code dimension growth with those derived in \cite{7282444}.

    The rank-1 approach establishes that the expected dimension of the star product of two uniformly random linear codes $\C_1,\C_2 \leq \F_q^n$ attains $\min\{\dim\C_1 \dim \C_2,n\}$ for a restricted parameter regime. Specifically, let $\varepsilon,\kappa\in(0,1)$ such that $q^{(1-\kappa)^2}\geq\frac{2q-1}{q}$, and define the parameter space
    \begin{equation}
        \mathcal{P}(\varepsilon,\kappa) = \left\{ (k_1,k_2) \in \mathbb{N}\times \mathbb{N} \;\middle|\; 2\leq k_1 \leq k_2 \leq \frac{\varepsilon q^{\kappa k_1}}{(q-1)k_1} \right\}. \label{eq:parameterspace}
    \end{equation}
    Under these conditions, \cite[Theorem 16]{7282444} yields the same conclusion as our Theorem~\ref{thm:expectedink} when $(\dim \C_1, \dim \C_2)\in \mathcal{P}(\varepsilon,\kappa)$. However, our result is stronger because any sequence of code dimensions in $\mathcal{P}(\varepsilon,\kappa)$ satisfies the conditions of Theorem~\ref{thm:expectedink}. Similarly, \cite[Theorem 17]{7282444} reaches the same conclusion as our Theorem~\ref{thm:expectedink2} when $(\dim \C_1, \dim \C_2)\in \mathcal{P}(\varepsilon,\frac{1}{2})$. In this specific case, the rank-1 approach of \cite{7282444} is stronger only for $q=2$, since $\alpha < 2 - \log_q(2q-1) < 0.415$. For $q\geq 3$, we have $2-\log_q(2q-1) \geq 0.535$, making our bound strictly stronger. Furthermore, our Corollary~\ref{cor:randomcodesinq}, which dictates the asymptotic behavior in the field size, remains entirely independent of the relationship between $\dim \C_1$ and $\dim \C_2$.
\end{remark}
}

{ 
We also consider the variance of the star product dimension to characterize its concentration behavior. In particular, Proposition~\ref{prop:variance_t} shows that if $n(t)<k_1(t)k_2(t)$ for all $t\in \NN$ and $n(t)$ grows polynomially, then the variance vanishes at least exponentially. Similarly, if $n(t)$ grows polynomially and $n(t)> k_1(t)k_2(t)$ for all $t\in \NN$, then the variance vanishes at least exponentially to $0$ provided that the dimension gap $n(t)-k_1(t)k_2(t)$ grows at least linearly in $t$.

\begin{proposition}\label{prop:variance_t}
    Let $k_1,k_2\colon \NN \rightarrow \NN$ be admissible and let $n\colon \NN\rightarrow \NN$ be a strictly monotone increasing function. Let $\C_1,\C_2\leq \F_q^{n(t)}$ be uniformly random codes of dimensions $k_1(t)$ and $k_2(t)$, respectively, as in Definition~\ref{def:setuprand}. Let $\gamma=\frac{2q-1}{q^2}$. Then
    $$
        \mathrm{Var}(\dim (\C_1\star \C_2)) = \begin{cases}
            \mathcal{O}_t\left(n(t)^2\gamma^t\right) & \text{if } n(t)<k_1(t)k_2(t), \\
            \mathcal{O}_t(1) & \text{if } n(t)=k_1(t)k_2(t), \\
            \mathcal{O}_t\left(n(t)^2\gamma^{n(t)-k_1(t)k_2(t)}\right) & \text{if } n(t)>k_1(t)k_2(t).
        \end{cases}
    $$
\end{proposition}
\begin{proof}
Let $M(t)=\min\{k_1(t)k_2(t),n(t)\}$ and consider the non-negative integer-valued random variable $X_t=M(t)-\dim (\C_1\star \C_2)$. Since $M(t)$ is deterministic for a fixed $t$, we have $\mathrm{Var}(\dim (\C_1\star \C_2)) = \mathrm{Var}(X_t)$. For the first and third cases, we bound the variance using $\mathrm{Var}(X_t) \leq \mathbb{E}[X_t^2] \leq n(t)^2\mathbb{P}(X_t\geq 1)$.

    Suppose $n(t)< k_1(t)k_2(t)$ for all $t\in \NN$. The event $X_t \geq 1$ implies $\dim(\C_1\star \C_2) <n(t)$. By Theorem~\ref{thm:expectedink2}, for sufficiently large $t$, we have $\mathbb{P}(X_t\geq 1)\leq \gamma^t$, which yields $\mathrm{Var}(\dim(\C_1\star \C_2)) = \mathcal{O}_t(n(t)^2 \gamma^t)$.

    Suppose $n(t)> k_1(t)k_2(t)$ for all $t\in \NN$. The event $X_t\geq 1$ is equivalent to $|\ker \psi_{\C_1,\C_2}|>1$. Applying Markov's inequality to the kernel size yields $\mathbb{P}(X_t\geq 1) \leq \mathbb{E}[|\ker \psi_{\C_1,\C_2}| - 1]$. By Theorem~\ref{thm:expectedink}, for sufficiently large $t$, this probability is bounded by $\gamma^{n(t)-k_1(t)k_2(t)}$, giving $\mathrm{Var}(\dim(\C_1\star \C_2)) = \mathcal{O}_t\left(n(t)^2 \gamma^{n(t)-k_1(t)k_2(t)}\right)$.

    Suppose $n(t) = k_1(t)k_2(t)$ for all $t\in \NN$. Let $K_t = |\ker \psi_{\C_1, \C_2}|$, so $K_t = q^{X_t}$. By Theorem~\ref{thm:asympkerneldimension2}, $\lim_{t \to \infty} \mathbb{E}[K_t] = 2$. Consequently, for sufficiently large $t$, the expectation is bounded by a non-zero constant $C$. By Markov's inequality, for any integer $j \geq 1$, the tail probability of the dimension defect for sufficiently large $t$ is bounded by 
    $$\mathbb{P}(X_t \geq j) = \mathbb{P}(K_t \geq q^j) \leq \frac{\mathbb{E}[K_t]}{q^j} \leq Cq^{-j}.$$
        Using this uniform exponential tail bound, the second moment satisfies 
        $$\mathbb{E}[X_t^2] = \sum_{j=1}^{\infty} (2j-1)\mathbb{P}(X_t \geq j) \leq \sum_{j=1}^{\infty} (2j-1) Cq^{-j}.$$
        Because $q \geq 2$, this series converges absolutely to a finite constant independent of $t$, so $\mathbb{E}[X_t^2] = \mathcal{O}_t(1)$. Since $\mathrm{Var}(X_t) \leq \mathbb{E}[X_t^2]$, it follows that $\mathrm{Var}(\dim(\C_1 \star \C_2)) = \mathcal{O}_t(1)$.
\end{proof}
\begin{remark}
Establishing tight lower bounds to completely characterize the variance's $\Omega$-behavior remains an open problem. While Proposition~\ref{prop:variance_q} and Proposition~\ref{prop:variance_t} provide explicit upper bounds on the variance, which vanish asymptotically only when the code length and the dimensions satisfy appropriate relative growth conditions, determining the exact decay rates from below requires a granular second-moment analysis of the kernel size. Unlike our first-moment results, which leverage the linearity of expectation over individual bilinear forms, computing the second moment necessitates finding the joint probability that pairs of forms vanish simultaneously. This introduces complex geometric dependencies that break linearity and prevent a straightforward extension of our current machinery.
\end{remark}
}

\section{{Implications to secure communications and quantum error correction}}
\label{sec:app}

As mentioned in the introduction, the star product appears in a variety of applications, including but not limited to cryptanalysis, linear exact repair schemes, quantum error correction, private information retrieval, and secure distributed matrix multiplication.  In this section, we briefly discuss some of these applications and the connection to our results.

\paragraph{Private information retrieval (PIR).}
Private information retrieval considers the problem of downloading a data item or file from a (public) database without disclosing the identity of the retrieved item to the database owner \cite{chor1995private}. Recently, PIR from coded distributed storage systems has gained a lot of attention, see  \cite{Sun2016,Banawan2018,PIR2017,skoglund2019pir} among many others. 

Star products are natural objects to study in linear  PIR and its quantum extensions due to the response structure, which is essentially just a vector of inner products between the query vectors and the stored data vectors. The utility of star products in PIR is showcased by several extensions of the basic setting:  the first star product PIR scheme \cite{PIR2017} has been extended to cover, e.g., stragglers and errors/adversarial nodes \cite{Tajeddine2018ByzStar}, non-MDS codes \cite{tPIR}, secret sharing framework for secure PIR with algebraic geometry codes \cite{makkonen2024secretsharingsecureprivate}, and quantum PIR \cite{allaix2020quantum}. It is also the key ingredient in proving and achieving the capacity of PIR from coded storage with colluding servers \cite{Holzbaur2019starcapa,holzbaur2022tit}. 

A figure of merit for any PIR scheme is its \emph{retrieval rate}, defined as 
$$R_{\mathrm{PIR}}=\frac{\mathrm{size\ of\ retrieved\ file}}{\mathrm{size\ of\ total\ download}}\leq 1,$$ which should be maximized. The star product retrieval scheme \cite{PIR2017} satisfies
\begin{align*}
   \frac{d(\C\star \D)-1}{n} \leq R_{\mathrm{PIR}} \leq \frac{ \dim ((\C\star \D)^\perp)}{n} = 1- \frac{\dim (\C \star \D)}{n}. \label{eq:pirrate}
\end{align*}

Suppose now that both $\C$ and $\D$ are chosen randomly (according to Definition~\ref{def:setuprand}) such that $\dim \C = k_1$ and $\dim \D = k_2$ . Then by applying Corollary \ref{cor:randomcodesinq} we see that
\[
    \lim_{q\rightarrow \infty} \mathbb{E}[R_{\mathrm{PIR}}] \leq 1 - \lim_{q\rightarrow \infty}\frac{ \mathbb{E}[\dim( \C \star \D)]}{n} = 1- \frac{\min\{k_1 k_2,n\}}{n}.
\]

Thus, when instantiating such a PIR scheme with random codes, a non-zero upper bound on the expected rate is obtained when $\dim \C \dim \D < n$. 
Now suppose $\dim \C_1 = k_1(t)$ and $\dim \C_2 = k_2(t)$ according to Corollary \ref{cor:mainresult}. Then we see
\[
    \mathbb{E}[R_{\mathrm{PIR}}] \leq 1 - \frac{\min\{k_1(t)k_2(t),n(t)\}}{n(t)} + o_t(1),
\]
so we obtain the same conclusion as for the asymptotics in $q$. Contrasting these with the star product PIR scheme utilizing GRS codes \cite{PIR2017} that achieves a rate 
$$
R_{\mathrm{PIR}}=1-\frac{k_1+k_2-1}{n},
$$
we see that random codes will typically not yield constructions with an asymptotically good PIR rate. 

{ 
As suggested by the empirical distribution in Table~\ref{tab:distribution_experiments}, over small fields such as $\F_2$, there is a non-negligible probability of randomly sampling code pairs that yield a sub-maximal star product dimension. This implies that one can experimentally search for and successfully identify random code pairs with favorable parameters for PIR schemes. While decoding random linear codes is generally computationally hard in the presence of Byzantine servers, this limitation is not prohibitive in the standard, error-free PIR setting. In the absence of malicious error injection or channel errors, the user retrieval process solely relies on erasure decoding equivalent to solving a linear system, and hence these randomly sampled codes remain practically applicable and efficient for retrieval.
}

\paragraph{Secure distributed matrix multiplication (SDMM).}

Secure distributed matrix multiplication is a protocol for computing a large matrix product using a set of helper servers, such that the content of the matrices remain information-theoretically secure against specified subsets of colluding servers. In \cite{SDMM2024}, a general framework for linear SDMM using star products was established and shown to capture most of the previous SDMM schemes in the literature as special cases, including \cite{d2020gasp, aliasgari2020private, mital2022secure, machado2023hera}. In this framework, a user holds two private matrices $A$ and $B$ and wishes to compute $AB$, while having access to $N$ servers. The user constructs encoded shares and distributes them to the servers using two codes $\C_A,\C_B\leq \F_q^N$. 

We define the \emph{recovery threshold} $R_{\mathrm{SDMM}}$ as the smallest number of server responses needed to guarantee reconstruction of $AB$, and the \emph{straggler tolerance} $S_{\mathrm{SDMM}}$ as the maximum number of unresponsive servers that can be tolerated while still able to reconstruct the product. For a linear SDMM scheme that is decodable, that is, able to reconstruct $AB$, there exist corresponding bounds on the recovery threshold and straggler tolerance.

Instantiating a linear SDMM scheme with random codes (according to Definition~\ref{def:setuprand}) $\C_A,\C_B\leq \F_q^N$ yields with high probability a decodable scheme for sufficiently large $q$. According to \cite{SDMM2024}, such a scheme, when decodable, has a recovery threshold of $R_{\mathrm{SDMM}} = N-d( \C_A \star \C_B)+1$ and can tolerate at most $S_{\mathrm{SDMM}} = d(\C_A \star \C_B)-1$ stragglers. Consequently,
\begin{align*}
    N \geq \lim_{q\rightarrow \infty} \mathbb{E}[R_{\mathrm{SDMM}}] &\geq \min\{\dim \C_A \dim \C_B,N\}
    \intertext{and} 
    \lim_{q\rightarrow \infty} \mathbb{E}[S_{\mathrm{SDMM}}] &\leq N-\min\{\dim \C_A \dim \C_B,N\}.
\end{align*}

{
 
Thus, if $\dim\C_A \dim \C_B < N$ we get a non-maximal lower bound on the expected recovery threshold and a non-zero upper bound on the expected straggler tolerance. However, if $\dim\C_A \dim \C_B \geq N$, then the expected recovery threshold asymptotically approaches $N$ and consequently the straggler tolerance vanishes asymptotically. While this indicates that random constructions typically do not possess the desirable properties of an asymptotically low recovery threshold and an asymptotically high straggler tolerance, the empirical distribution in Table~\ref{tab:distribution_experiments} shows, similarly to the PIR setting, that over small fields there is a non-negligible probability of sampling code pairs with sub-maximal star products. This implies that there is still practical space to experimentally search for and successfully identify specific random code pairs that behave well for SDMM.
}

\paragraph{Binary CSS-T quantum error-correcting codes.}
Quantum error-correcting codes are essential for protecting quantum information against noise and decoherence. A particularly important class is the Calderbank–Shor–Steane (CSS) family, where two classical binary linear codes with specific inclusion properties allow the construction of a quantum code \cite{CSS1,CSS2}. In the setting of transversal T gates, further structural conditions are imposed on the pair of classical codes \cite{CSST}. The resulting CSS-T codes can then be characterized via the star product \cite{camps2023algebraic}. In particular, binary CSS-T codes are pairs of binary non-zero linear codes $(\C_1,\C_2)$ satisfying $\C_2 \subseteq \C_1 \cap (\C_1 \star \C_1)^\perp.$ Furthermore, the resulting binary CSS-T code will have a minimum distance $d'\geq d(\C_2^\perp)$.

For a uniformly random code $\C_1\leq \F_2^n$ with $\dim \C_1 = k$, it was shown in \cite{cascudo2015squares} that for sufficiently large $k$, 
\[
    \dim (\C_1 \star \C_1 )= \min\left\{ n, \binom{\dim \C_1+1}{2}\right\}
\]
 with high probability. Consequently,
\[
    \dim ((\C_1\star \C_1)^\perp )= n  - \min\left\{ n,\binom{\dim \C_1+1}{2}\right\},
\]
Thus, for a uniformly random code $\C_1 \le \F_2^n$ of sufficiently large dimension, we expect with high probability that no binary CSS-T code using $\C_1$ exists. 
When $\dim(\C_1 \star \C_1)$ is large but still smaller than the ambient space, as is typical for random codes when $\dim \C_1$ is not too large, the orthogonal complement $(\C_1 \star \C_1)^\perp$ becomes small. As a result, $\C_2$ is forced to have low dimension. Consequently, the resulting CSS-T code offers only a weak lower bound on the minimum distance, $d' \ge d(\C_2^\perp)$. Conversely, codes with a smaller star product leave more room for choices of $\C_2$ and may yield quantum codes with larger distance.

\section{Conclusions and future work}

We have shown that the expected dimension of the star product of two uniformly random linear codes asymptotically reaches its maximal value, unconditionally in the field size and, under appropriate assumptions, also in the code dimensions. It would be interesting to analyze the non-asymptotic cases further and to derive lower bounds for the expected dimension and its variance, thereby providing a more detailed understanding on the distribution of the probability mass. 

The behavior of star products underpins the performance of many applications such as private information retrieval and secure distributed matrix multiplication, which we also briefly discussed in this paper. One could keep exploring how the derived star product properties extend to other settings where star products play a central role. In particular, the results suggest potential applications to linear exact repair schemes for distributed storage, where componentwise operations naturally arise \cite{guruswamiwooters,matthewsrepair}, and to code-based cryptanalysis, where deviations from maximal star product dimension can serve as structural distinguishers. For instance, in \cite{Mora2022dual} it is shown that the dimension of the square of the dual of a Goppa code is significantly smaller than that of a random code, providing a rigorous upper bound for its dimension. However, the findings concern parameter ranges that do not threaten the current  McEliece Goppa code system parameters, see \cite[Table 1]{Mora2022dual}. It would be interesting to try using different codes in the star product instead of the Goppa code or its dual itself. This would enable changing the dimension of the other code while keeping the Goppa code dimension fixed, which could potentially reveal some differences versus the same code starred with a random code in new parameter ranges relevant to cryptography.

\paragraph*{Acknowledgement.} The authors would like to thank Dr. 
Makkonen and Neehar Verma for valuable discussions and for Example~\ref{exmp:mdscodes}.

\bibliographystyle{IEEEtran}
\bibliography{star-biblio}

\pagebreak 
\begin{appendices}

\section{Comparison of expected dimension and Corollary \ref{cor:lowerboundstar}} \label{app:table}

\begin{table}[!ht]
\footnotesize
\centering
\captionsetup{width=0.95\textwidth}
\caption{{Comparison between $\hat{\mathbb{E}}[\dim(\mathcal{C}_1 \star \mathcal{C}_2)]$ and the lower bound from Corollary~\ref{cor:lowerboundstar}. {Here, $\hat{\mathbb{E}}[\dim(\mathcal{C}_1 \star \mathcal{C}_2)]$ denotes the estimate of $\mathbb{E}[\dim(\mathcal{C}_1 \star \mathcal{C}_2)]$ obtained by a Monte Carlo simulation with 100,000 samples per instance.} 
}}
\label{tab:lowerboundexperiments}
\begin{tabular}{|c|c|c|c|c|c|c|c|}
\hline
$n$                & $k_1$              & $k_2$              & $q$ & $\hat{\mathbb{E}}[\dim (\C_1\star\C_2)]$ & Corollary \ref{cor:lowerboundstar} & Ratio \\
\hline
\multirow{4}{*}{7}  & \multirow{4}{*}{2} & \multirow{4}{*}{3} & 2   & 4.6300 & 4.3629 & 1.06122 \\
                    &                    &                    & 3   & 5.4378 & 5.1610 & 1.05363 \\
                    &                    &                    & 5   & 5.8495 & 5.6761 & 1.03055 \\
                    &                    &                    & 7   & 5.9413 & 5.8348 & 1.01825 \\
\hline
\multirow{4}{*}{7}  & \multirow{4}{*}{2} & \multirow{4}{*}{4} & 2   & 5.2775 & 4.9424 & 1.06780 \\
                    &                    &                    & 3   & 6.1982 & 5.8168 & 1.06557 \\
                    &                    &                    & 5   & 6.7174 & 6.4143 & 1.04725 \\
                    &                    &                    & 7   & 6.8601 & 6.6309 & 1.03457 \\
\hline
\multirow{4}{*}{7}  & \multirow{4}{*}{3} & \multirow{4}{*}{3} & 2   & 5.7126 & 5.4339 & 1.05129 \\
                    &                    &                    & 3   & 6.5438 & 6.2843 & 1.04129 \\
                    &                    &                    & 5   & 6.8996 & 6.7708 & 1.01902 \\
                    &                    &                    & 7   & 6.9643 & 6.8982 & 1.00958 \\
\hline
\multirow{4}{*}{7}  & \multirow{4}{*}{3} & \multirow{4}{*}{4} & 2   & 6.1892 & 5.9594 & 1.03856 \\
                    &                    &                    & 3   & 6.7806 & 6.6232 & 1.02377 \\
                    &                    &                    & 5   & 6.9594 & 6.9011 & 1.00845 \\
                    &                    &                    & 7   & 6.9869 & 6.9582 & 1.00412 \\
\hline
\multirow{4}{*}{11} & \multirow{4}{*}{2} & \multirow{4}{*}{3} & 2   & 5.5320 & 5.3628 & 1.03155 \\
                    &                    &                    & 3   & 5.9525 & 5.9117 & 1.00690 \\
                    &                    &                    & 5   & 5.9983 & 5.9960 & 1.00038 \\
                    &                    &                    & 7   & 5.9998 & 5.9996 & 1.00003 \\
\hline
\multirow{4}{*}{11} & \multirow{4}{*}{2} & \multirow{4}{*}{4} & 2   & 6.9100 & 6.5735 & 1.05119 \\
                    &                    &                    & 3   & 7.7995 & 7.6387 & 1.02105 \\
                    &                    &                    & 5   & 7.9860 & 7.9639 & 1.00277 \\
                    &                    &                    & 7   & 7.9978 & 7.9930 & 1.00060 \\
\hline
\multirow{4}{*}{11} & \multirow{4}{*}{3} & \multirow{4}{*}{3} & 2   & 7.6656 & 7.3205 & 1.04714 \\
                    &                    &                    & 3   & 8.7157 & 8.5237 & 1.02253 \\
                    &                    &                    & 5   & 8.9741 & 8.9360 & 1.00426 \\
                    &                    &                    & 7   & 8.9943 & 8.9822 & 1.00135 \\
\hline
\multirow{4}{*}{11} & \multirow{4}{*}{3} & \multirow{4}{*}{4} & 2   & 8.9683 & 8.5278 & 1.05165 \\
                    &                    &                    & 3   & 10.338 & 9.9850 & 1.03535 \\
                    &                    &                    & 5   & 10.861 & 10.691 & 1.01590 \\
                    &                    &                    & 7   & 10.948 & 10.851 & 1.00894 \\
\hline
\multirow{4}{*}{15} & \multirow{4}{*}{2} & \multirow{4}{*}{3} & 2   & 5.8537 & 5.7877 & 1.01140 \\
                    &                    &                    & 3   & 5.9960 & 5.9922 & 1.00063 \\
                    &                    &                    & 5   & 6.0000 & 5.9990 & 1.00017 \\
                    &                    &                    & 7   & 6.0000 & 6.0000 & 1.00000 \\
\hline
\multirow{4}{*}{15} & \multirow{4}{*}{2} & \multirow{4}{*}{4} & 2   & 7.6326 & 7.4760 & 1.02095 \\
                    &                    &                    & 3   & 7.9843 & 7.9699 & 1.00181 \\
                    &                    &                    & 5   & 7.9999 & 7.9996 & 1.00004 \\
                    &                    &                    & 7   & 8.0000 & 8.0000 & 1.00000 \\
\hline
\multirow{4}{*}{15} & \multirow{4}{*}{3} & \multirow{4}{*}{3} & 2   & 8.5634 & 8.3906 & 1.02059 \\
                    &                    &                    & 3   & 8.9803 & 8.9642 & 1.00180 \\
                    &                    &                    & 5   & 8.9998 & 8.9995 & 1.00003 \\
                    &                    &                    & 7   & 9.0000 & 9.0000 & 1.00000 \\
\hline
\multirow{4}{*}{15} & \multirow{4}{*}{3} & \multirow{4}{*}{4} & 2   & 10.850 & 10.473 & 1.03600 \\
                    &                    &                    & 3   & 11.885 & 11.793 & 1.00780 \\
                    &                    &                    & 5   & 11.996 & 11.990 & 1.00050 \\
                    &                    &                    & 7   & 11.999 & 11.998 & 1.00008 \\
\hline
\end{tabular}
\end{table}

\begin{table}[!ht]
 \footnotesize
\centering
\captionsetup{width=0.95\textwidth}
\caption{Empirical distribution and sample variance ($\hat{\sigma}^2$) of the star product dimension. This data is derived from the exact same 100,000 Monte Carlo samples per parameter instance as in Table~\ref{tab:lowerboundexperiments}. The `Range' column specifies the interval $[\kappa_{\min}, \kappa_{\max}]$ of dimensions observed at least once. The `Frequencies' column provides the tuple of exact observation counts $(f_{\kappa_{\min}}, f_{\kappa_{\min}+1}, \dots, f_{\kappa_{\max}})$.}
\label{tab:distribution_experiments}
\begin{tabular}{|c|c|c|c|c|l|c|}
\hline
$n$ & $k_1$ & $k_2$ & $q$ & Range & Frequencies & $\hat{\sigma}^2$ \\
\hline
\multirow{4}{*}{7}  & \multirow{4}{*}{2} & \multirow{4}{*}{3} 
                    & 2   & $[2, 6]$ & $(665, 8404, 32721, 43684, 14526)$ & 0.73161 \\
                    &                    &                    & 3   & $[2, 6]$ & $(18, 479, 7316, 40077, 52110)$ & 0.42336 \\
                    &                    &                    & 5   & $[3, 6]$ & $(8, 516, 13989, 85487)$ & 0.13862 \\
                    &                    &                    & 7   & $[4, 6]$ & $(74, 5725, 94201)$ & 0.05676 \\
\hline
\multirow{4}{*}{7}  & \multirow{4}{*}{2} & \multirow{4}{*}{4} 
                    & 2   & $[2, 7]$ & $(206, 3128, 16505, 37539, 34117, 8505)$ & 0.91309 \\
                    &                    &                    & 3   & $[2, 7]$ & $(4, 120, 1938, 14373, 45120, 38445)$ & 0.57787 \\
                    &                    &                    & 5   & $[3, 7]$ & $(2, 93, 2070, 23832, 74003)$ & 0.24996 \\
                    &                    &                    & 7   & $[4, 7]$ & $(10, 483, 12993, 86514)$ & 0.13058 \\
\hline
\multirow{4}{*}{7}  & \multirow{4}{*}{3} & \multirow{4}{*}{3} 
                    & 2   & $[3, 7]$ & $(572, 7584, 30128, 43447, 18269)$ & 0.75620 \\
                    &                    &                    & 3   & $[3, 7]$ & $(6, 315, 5379, 33893, 60407)$ & 0.37529 \\
                    &                    &                    & 5   & $[4, 7]$ & $(4, 257, 9517, 90222)$ & 0.09572 \\
                    &                    &                    & 7   & $[5, 7]$ & $(39, 3490, 96471)$ & 0.03518 \\
\hline
\multirow{4}{*}{7}  & \multirow{4}{*}{3} & \multirow{4}{*}{4} 
                    & 2   & $[3, 7]$ & $(120, 2286, 15548, 42645, 39401)$ & 0.61594 \\
                    &                    &                    & 3   & $[4, 7]$ & $(44, 1446, 18911, 79599)$ & 0.20280 \\
                    &                    &                    & 5   & $[5, 7]$ & $(62, 3943, 95995)$ & 0.04025 \\
                    &                    &                    & 7   & $[5, 7]$ & $(4, 1304, 98692)$ & 0.01302 \\
\hline
\multirow{4}{*}{11} & \multirow{4}{*}{2} & \multirow{4}{*}{3} 
                    & 2   & $[2, 6]$ & $(23, 537, 6271, 32552, 60617)$ & 0.40938 \\
                    &                    &                    & 3   & $[3, 6]$ & $(2, 137, 4473, 95388)$ & 0.04813 \\
                    &                    &                    & 5   & $[4, 6]$ & $(1, 166, 99833)$ & 0.00169 \\
                    &                    &                    & 7   & $[5, 6]$ & $(16, 99984)$ & 0.00016 \\
\hline
\multirow{4}{*}{11} & \multirow{4}{*}{2} & \multirow{4}{*}{4} 
                    & 2   & $[2, 8]$ & $(5, 78, 931, 6071, 22444, 41752, 28719)$ & 0.84390 \\
                    &                    &                    & 3   & $[4, 8]$ & $(2, 89, 1541, 16694, 81674)$ & 0.19671 \\
                    &                    &                    & 5   & $[5, 8]$ & $(1, 11, 1378, 98610)$ & 0.01411 \\
                    &                    &                    & 7   & $[6, 8]$ & $(1, 215, 99784)$ & 0.00218 \\
\hline
\multirow{4}{*}{11} & \multirow{4}{*}{3} & \multirow{4}{*}{3} 
                    & 2   & $[3, 9]$ & $(3, 118, 1492, 9194, 29157, 40969, 19067)$ & 0.89209 \\
                    &                    &                    & 3   & $[5, 9]$ & $(5, 126, 2394, 23244, 74231)$ & 0.25952 \\
                    &                    &                    & 5   & $[6, 9]$ & $(1, 16, 2551, 97432)$ & 0.02557 \\
                    &                    &                    & 7   & $[7, 9]$ & $(1, 574, 99425)$ & 0.00574 \\
\hline
\multirow{4}{*}{11} & \multirow{4}{*}{3} & \multirow{4}{*}{4} 
                    & 2   & $[4, 11]$ & $(8, 129, 1350, 7274, 22258, 36562, 26599, 5820)$ & 1.15550 \\
                    &                    &                    & 3   & $[6, 11]$ & $(3, 77, 1173, 10523, 41349, 46875)$ & 0.51432 \\
                    &                    &                    & 5   & $[8, 11]$ & $(11, 423, 13037, 86529)$ & 0.12892 \\
                    &                    &                    & 7   & $[8, 11]$ & $(1, 53, 5141, 94805)$ & 0.05086 \\
\hline
\multirow{4}{*}{15} & \multirow{4}{*}{2} & \multirow{4}{*}{3} 
                    & 2   & $[2, 6]$ & $(1, 20, 887, 12792, 86300)$ & 0.14396 \\
                    &                    &                    & 3   & $[4, 6]$ & $(1, 398, 99601)$ & 0.00400 \\
                    &                    &                    & 5   & $[5, 6]$ & $(1, 99999)$ & $0.00001$ \\
                    &                    &                    & 7   & $[6, 6]$ & $(100000)$ & 0.00000 \\
\hline
\multirow{4}{*}{15} & \multirow{4}{*}{2} & \multirow{4}{*}{4} 
                    & 2   & $[4, 8]$ & $(46, 580, 4628, 25563, 69183)$ & 0.36531 \\
                    &                    &                    & 3   & $[5, 8]$ & $(1, 37, 1493, 98469)$ & 0.01625 \\
                    &                    &                    & 5   & $[7, 8]$ & $(14, 99986)$ & 0.00014 \\
                    &                    &                    & 7   & $[7, 8]$ & $(1, 99999)$ & $0.00001$ \\
\hline
\multirow{4}{*}{15} & \multirow{4}{*}{3} & \multirow{4}{*}{3} 
                    & 2   & $[4, 9]$ & $(2, 34, 690, 5858, 29727, 63689)$ & 0.40902 \\
                    &                    &                    & 3   & $[7, 9]$ & $(25, 1912, 98063)$ & 0.01973 \\
                    &                    &                    & 5   & $[8, 9]$ & $(22, 99978)$ & 0.00022 \\
                    &                    &                    & 7   & $[8, 9]$ & $(1, 99999)$ & $0.00001$ \\
\hline
\multirow{4}{*}{15} & \multirow{4}{*}{3} & \multirow{4}{*}{4} 
                    & 2   & $[6, 12]$ & $(11, 151, 1370, 7238, 23331, 40317, 27582)$ & 0.92638 \\
                    &                    &                    & 3   & $[8, 12]$ & $(2, 14, 511, 10432, 89041)$ & 0.11311 \\
                    &                    &                    & 5   & $[10, 12]$ & $(1, 441, 99558)$ & 0.00443 \\
                    &                    &                    & 7   & $[11, 12]$ & $(85, 99915)$ & 0.00085 \\
\hline
\end{tabular}
\end{table}

\end{appendices}

\end{document}